\documentclass[11pt] {article}
\usepackage[a4paper, margin=1.25in]{geometry} 
\usepackage[usenames,dvipsnames,svgnames,x11names]{xcolor}

\usepackage[pdftex]{graphicx}
\graphicspath{{./figures}}
\usepackage[cmex10]{amsmath}
\usepackage{amsthm}
\newtheorem{lemma}{Lemma}
\newtheorem{theorem}{Theorem}
\usepackage[ruled,vlined]{algorithm2e}
\usepackage{algorithmic}
\usepackage{multirow}
\usepackage{colortbl}
\usepackage{tablefootnote} 
\usepackage[pdftex,colorlinks=true,urlcolor=blue,citecolor=blue]{hyperref}
\usepackage{float}  
\usepackage{subcaption}  
\usepackage{braket}
\usepackage{booktabs} 
\usepackage{mathptmx}
\usepackage{cite}
\usepackage[toc,page]{appendix}
\usepackage{amsfonts}
\usepackage{comment}
\usepackage{epsfig}
\usepackage{dcolumn}
\usepackage{latexsym}

\def\orcid#1{\kern .08em\href{https://orcid.org/#1}{\includegraphics[keepaspectratio,width=0.7em]{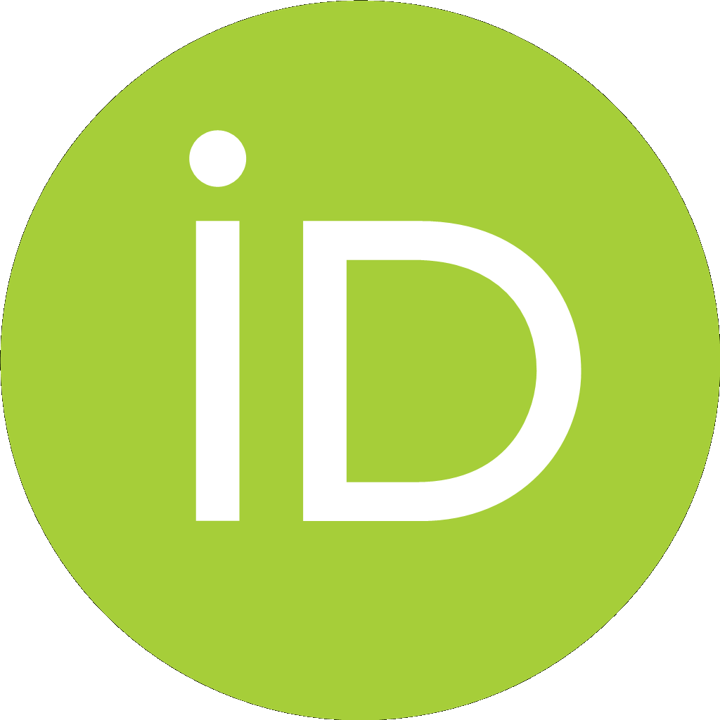}}}
\begin{document}

\title{Lightweight Targeted Estimation of Layout Noise in a Quantum Computer using Quality Indicator Circuits
\thanks{To appear in Quantum Information Processing (QIP), Springer-Nature, \href{https://doi.org/10.1007/s11128-026-05326-7}{doi:10.1007/s11128-026-05326-7}}
}

\author{Shikhar Srivastava$^{1}$, Ritajit Majumdar$^2$, Padmanabha Venkatagiri Seshadri$^2$,\\ Anupama Ray$^2$ and Yogesh Simmhan$^1$\orcid{0000-0003-4140-7774}\\~\\
\small $^1$Indian Institute of Science, Bangalore, India\\
\small $^2$IBM, Bangalore, India\\
\small sshikhar@iisc.ac.in, majumdar.ritajit@ibm.com, seshapad@in.ibm.com,\\ \small anupamar@in.ibm.com, simmhan@iisc.ac.in
}
\date{}

\maketitle

\begin{abstract}
Minimizing noise is essential for reliably executing quantum circuits on the current generation of hardware. The mapping of an abstract quantum circuit to the physical layout of the quantum hardware significantly influences the output quality since hardware noise profiles are non-uniform and dynamic. Existing solutions such as Mapomatic relies on stale calibrations while Just-In-Time (JIT) Transpilation has high hardware usage. We propose Quality Indicator Circuits (QICs) as a lightweight probe circuit synthesized from the user circuit, whose ideal noiseless outcome is known. QIC helps decide the region of the quantum hardware best suited for the given circuit. 
We propose QIC execution for each isomorphic layout, and an alternative approach where a union of multiple layouts are executed together for lower overheads. Our results show QIC to outperform Mapomatic in the quality of layout selection while having $79\%$ lower hardware overheads than JIT. This makes our QIC method lightweight, reliable and a viable technique for layout selection in near-term quantum devices.
\end{abstract}

\section{Introduction}
\label{sec:intro}

A general quantum-computational workflow is divided into several steps: mapping a problem of interest into quantum circuits and observables, optimizing the circuit for the target hardware, execution of the circuit on a quantum computer, and postprocessing the outcomes of the quantum computer to the desired solution~\cite{qiskit-pattern,9302814,beisel2022metamodel}. We focus on circuit optimization for a quantum hardware, often called \emph{transpilation}, which includes
mapping the virtual qubits of the circuit to the physical qubits of the hardware (referred to as the layout) to minimize the number of SWAP gates~\cite{10.1145/2431211.2431220,10.1145/3316781.3317859}. Multiple layouts can lead to the same number of SWAP gates for the virtual circuit. In such a case, the layout with the lowest noise profile is selected.

Finding the layout with the lowest noise profile is non-trivial. In general, for isomorphic circuit layouts, each layout is scored using \textit{calibration data} available for that hardware, based on a scoring algorithm to select the one with the best score~\cite{nation2023suppressing,bhoumik2025distributedschedulingquantumcircuits,murali2019noiseadaptivecompilermappingsnoisy}.
Among these, \textit{Mapomatic}~\cite{nation2023suppressing} used in IBM Quantum hardware is lightweight, since it does not require additional circuit-specific hardware execution overheads. The isomorphic layouts are found using classical computing, and the calibration data on the IBM Quantum devices is made available. However, the noise in quantum devices is not static while the calibration data is usually updated only every 3--24~hours due to the associated overheads, and can be outdated when they are used to calculate the noise profile~\cite{10.1145/3297858.3304007}. Therefore, obtaining the \textit{current} noise profile is necessary to decide the best layout for circuit execution.
This issue  
has been addressed through \emph{Just-in-time (JIT) transpilation}~\cite{wilson2020just} 
that runs a set of hardware characterization experiments to obtain updated information on the noise profile before deciding the optimal layout. This overcomes the staleness drawback of Mapomatic, but has significant hardware overheads for running the profiling experiments since it attempts to profile build a noise profile for each qubit or 2-qubit connectivity of the hardware, irrespective of the size of the user's circuit of interest. Experimental reports~\cite{ibm_qiskit_qubit_selection} show that, on average, 4~minutes of hardware time is required to obtain the updated noise profile using JIT for a 127-qubit IBM Quantum device.

In this article, we propose a novel targeted approach based on a \emph{Quality Indicator Circuit} (QIC)~\cite{defensive_publication}, which is an intermediate method between Mapomatic and JIT. Here, the noise profile for a layout is obtained by executing a QIC, which is uniquely constructed for the circuit given by the user such that: (i) it respects the structure of the user's circuit being deployed, and (ii) whose ideal outcome is either known or can be easily computed irrespective of its size. Upon execution of the QIC, the deviation of the obtained results from the ideal one indicates the quality of the layout. Unlike JIT, which evaluates the noise profile for all qubits or 2-qubit connectivity, the QIC provides an overall quality estimate of the \textit{specific layout} of the user's circuit, i.e., QIC can be considered as a targeted \emph{quality detection} mechanism for a given layout. This ensures that the hardware execution time for QIC is significantly lower than JIT, by focusing on specific layouts for the user's circuit, while still overcoming the staleness drawback of Mapomatic.

Since QIC can only detect the quality of a layout, and not directly provide the optimal layout for the circuit execution, we need to execute it on all the isomorphic layouts and then select best layout. Although this brute-force approach (\textit{Basic QIC}) often requires fewer quantum resource than JIT, it is not ideal. We propose two improved methods to reduce the number of QIC executions by uniting multiple QICs for different layouts to form a larger QIC that is executed on hardware, followed by marginalization to obtain the outcome of individual layouts. QICs inside the larger QIC can have no overlap or some overlap. In our first alternative approach, we create a single union QIC for multiple layouts with no overlaps to halve the number of circuit executions required, relative to smaller QICs (\textit{Union QIC}). However, for larger user circuits, it is difficult to obtain layouts with no overlap, and this method fails to reduce the number of circuit execution below the brute-force technique. So, a second alternative allows overlap between layouts, by maintaining a \emph{distortion threshold} (\textit{Overlap QIC}). 
Our results show that this outperforms Mapomatic in the quality of layout selection while reducing the hardware overhead of JIT by $79.7\%$ on average.
This makes our proposed targeted method lightweight and reliable, and a viable technique for layout selection in near-term quantum devices.

The rest of the article is arranged as follows: In Sec.~\ref{sec:background}, we discuss the background necessary for the paper and contrast related literature against our work; in Sec.~\ref{sec:qic}, we discuss the construction of the Quality Indicator Circuit (QIC) for a given user circuit, and demonstrate its effectiveness to evaluate the quality of layouts; 
in Sec.~\ref{sec:union_disjoint_qic}, we propose the \textit{Union QIC} method to create a single union QIC for multiple layouts; 
in Sec.~\ref{sec:distortion}, we improve this to with the \textit{Overlap QIC} technique that allows some overlap between layouts, limited by a distortion threshold;
in Sec.~\ref{sec:result}, we report hardware results that show that QIC is able to outperform Mapomatic in selecting a better layout, providing expectation values closer to the ideal than the layout selected by Mapomatic; in Sec.~\ref{sec:limitations}, we discuss other limitations and future extensibility of the proposed approach; and finally, in Sec.~\ref{sec:conclusion}, we offer our conclusions on our technique for layout selection in near-term quantum devices.

\section{Background and Related Work}
\label{sec:background}
 
Quantum workflows can be partitioned into four broad steps~\cite{qiskit-pattern,9302814,beisel2022metamodel} 
-- (i) \textit{Mapping} the problem into quantum circuits and operators, (ii) \textit{Optimizing} the circuits and operators for the target hardware, (iii) \textit{Executing} the circuits on the target hardware, and (iv) \textit{Post-processing} the result into the required solution. 
Step (ii), also called \textit{transpilation} comprises of different techniques to optimally map the circuit to the target hardware. 
Here, an abstract quantum circuit is mapped to the constraints of the target hardware. These constraints come from (i) the library of the basis gate set for that hardware~\cite{ lin2013optimized, niemann2016logic, lin2014paqcs}, (ii) mapping of the abstract qubits to the underlying topology of the hardware coupling map~\cite{ zou2024lightsabre, zulehner2018efficient, murali2019noise, murali2020software, sivarajah2020t}, and (iii) selecting the qubits and connectivity with low noise~\cite{ murali2019noise, murali2020software, nation2023suppressing, wilson2020just}. In this paper, we address the challenge of selecting of qubits and connectivity for the circuit layout with low noise profile.

\subsection{SWAP Gate Minimization Approaches} 
In superconducting quantum computers, limited qubit connectivity often necessitates the insertion of SWAP gates to implement CNOT operations between non-adjacent qubits, for a quantum circuit where each SWAP gate consists of 3 CNOT gates.
Minimizing the number of SWAP gates is an NP-hard problem~\cite{zhu2022complexity}.


Several approaches primarily minimize SWAP gates as part of circuit layout optimization~\cite{shaik2023optimallayoutsynthesisquantum,10247760,shaik2024optimallayoutsynthesisdeep}. 
Shaik and van de Pol~\cite{shaik2023optimallayoutsynthesisquantum} reduce SWAP overhead for optimum layout synthesis motivated by the fact that fewer SWAP gate counts will reduce the total gate counts and hence, the total error rate. But they do not take the error rates of qubits and gates into consideration. However, selecting layout with least noise is equally important since there might be many isomorphic layouts with different noise profiles. In contrast to these approaches, 
our work emphasizes the selection of layout with least noise. 
While Murali et al.~\cite{murali2019noise} try to tackle both SWAP and layout noise optimization together, the problem, in some sense, becomes even \emph{harder}. They propose two greedy strategies for layout mappings -- one for selecting the region of hardware prominent in lower CNOT and qubit error rates, and another for selecting qubit pairs having frequent CNOT gates to reduce SWAP overhead. They used individual qubit and gate error rates obtained from calibration data to evaluate the overall noise of the layouts. In our QIC approach, we evaluate the layout noise globally, in real time without worrying about individual qubit and gate error rates.

\subsection{Mapomatic Approach}
Mapomatic~\cite{nation2023suppressing} performs this layout optimization in two steps. First, the noise profile of the hardware is ignored, and the algorithm tries to minimize the number of SWAP gates during the mapping. Once a layout $l$ with minimal SWAP count is obtained, it finds a list of layouts which are isomorphic to $l$. The algorithm then \emph{scores} each of the layout as $s = 1-\Pi_{x \in \{gate, measurement\}}(1-e_x)$, where $e_x$ denotes the probability of error of $x \in \{gate, measurement\}$. By construction, the lower the score, the better is the quality of the layout. This two-step approach makes Mapomatic better compared to previous efforts to optimize for both SWAP count and noise profile, and it is currently integrated in the IBM Quantum SDK Qiskit~\cite{javadi2024quantum}.

Mapomatic extracts the noise information $e_x$ from the backend hardware calibration data, and therefore requires no additional quantum overhead for scoring the layout for each user circuit. However, the calibration data itself is updated only once every $3-24$ hours for IBM Quantum devices. Since noise in quantum devices is not static, Mapomatic can end up providing incorrect scoring because it uses outdated noise information. 
Further, we demonstrate that our QIC scoring yields better layouts than Mapomatic for small width circuit while for large width circuits, both are comparable. However, unlike Mapomatic, which incurs no additional hardware time but uses outdated error rates from calibration data, QIC scoring is done in real time but requires extra hardware time depending, on the efficiency of circuit reduction obtained via \textit{Overlap QIC} technique.

\subsection{JIT Approaches} 
JIT transpilation~\cite{wilson2020just} overcomes the drawback of stale calibration by running benchmarking experiments to obtain the updated noise information. These experiments include SPAM characterization, $T_1$ and $T_2$ characterization~\cite{krantz2019quantum, nation2021scalable}, and randomized benchmarking for single and two-qubit gates~\cite{silva2025hands, magesan2012characterizing}. While characterization of SPAM and $T_1, T_2$ are quite lightweight, randomized benchmarking requires a lot of circuit executions -- we show later that 132 circuit executions are required for a 27 qubit IBM Quantum device. 

More generally, JIT takes a fixed QPU time irrespective of the user circuit since it characterizes the noise profile for all qubits in the hardware. This has the benefit for not having to be repeated for different user circuits, but carries high overheads if not amortized for the optimization of many circuits before it too grows stale like the hardware calibration data. Therefore, while JIT improves over Mapomatic by using more recent noise information, it has a significant quantum overhead.

In contrast, our QIC scoring depends on the specific circuit of interest. Different QICs are generated and executed for user circuits with different 2-qubit gate structures, and the corresponding execution time is proportional to the number of isomorphic layouts, and reduced further for the united approaches.
QIC is 
well-suited, and the profiling results can be reused, if user circuits with the same 2-qubit gate structure are executed frequently. For example, variational algorithms~\cite{Peruzzo_2014} execute the same circuit structure over multiple iterations. In such a scenario, using a $QIC$ at the beginning, or at regular intervals if the iterations extends for a significant time, can be cheaper than JIT. But its results are not generalizable to any user circuit, unlike JIT. However, given its lower execution cost, it is viable to generated and execute such targeted QIC for individual user circuits, on-demand.

In summary, we propose a middle-ground of the two approaches where we execute a lightweight QIC targeted to the user provided circuit, and obtain an updated global information about the noise profile relevant only to the proposed layout. 

\section{Basic Quality Indicator Circuit (QIC) for Layout Scoring}
\label{sec:qic}

\subsection{Empirical Motivation}
\begin{figure}[t]
    \centering  
    \includegraphics[width=0.5\columnwidth]{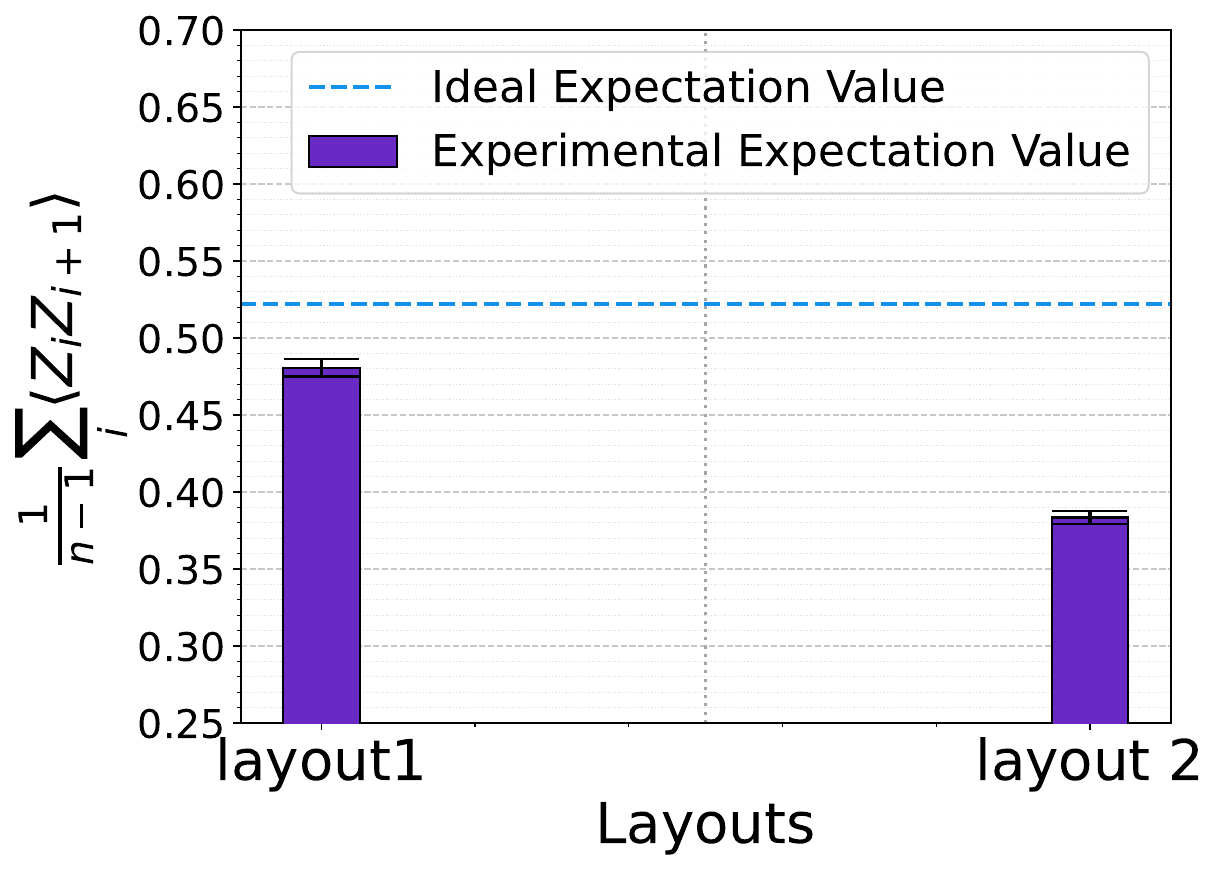}  
    \caption{Expectation values for a \textit{6-qubit QAOA circuit} executed on 2 different isomorphic layouts of a 27-qubit \textit{Fake Backend}. The expectation values differ significantly due to the differing noise profiles of the layouts even though the circuits executed on the two layouts are identical. 
    }
    \label{fig:why_we_need_best_layout}
\end{figure}

\begin{figure}[t]
    \centering  
    \includegraphics[width=0.5\columnwidth]{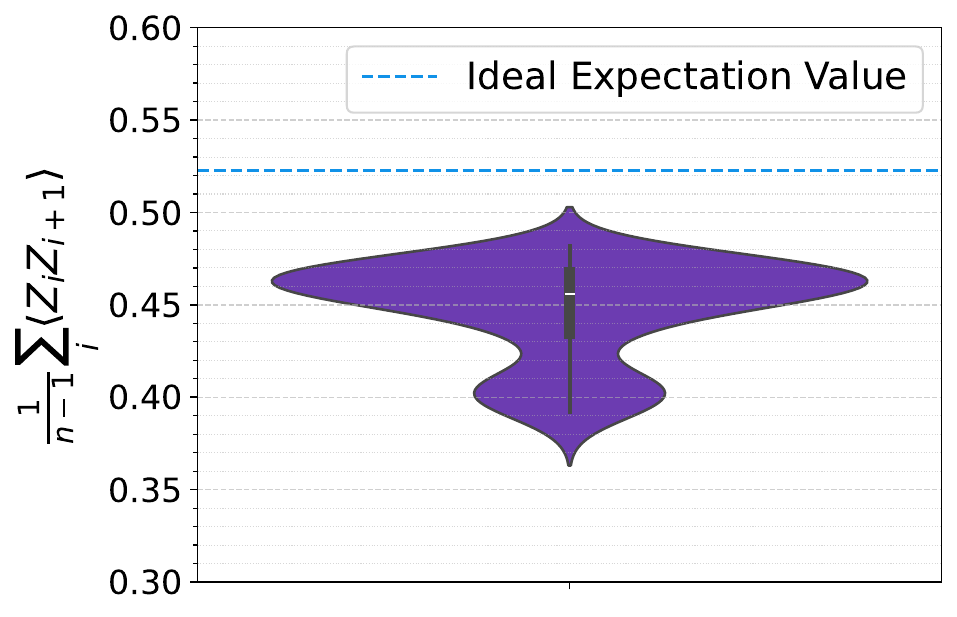}  
    \caption{Distribution of expectation values for 104 different isomorphic layouts of a \textit{6-qubit QAOA circuit} executed on a 27-qubit \textit{Fake Backend} with different noise profiles. The Q1, Q2 (median) and Q3 quartiles span a wide range of $0.435$, $0.456$ and $0.467$.
    }
    \label{fig:why_we_need_best_layout_violin}
\end{figure}

Mapping a virtual quantum circuit to the underlying hardware requires \textit{synthesis} of a functionally equivalent circuit created only using the basis gate set, which is the native gate library for a given hardware~\cite{2405.13196, camps2020approximate, amy2013meet}; \textit{mapping} the virtual circuit to the underlying hardware topology map~\cite{li2019tackling, sivarajah2020t, zou2024lightsabre}; and finding the \textit{isomorphic layout} with the least noise profile~\cite{nation2023suppressing, wilson2020just}. 
In Fig.~\ref{fig:why_we_need_best_layout} we motivate the importance of finding a layout with low noise profile. We compute the expectation value $\frac{1}{n-1}\sum\limits_{i=1}^{n-1}\langle Z_i Z_{i+1} \rangle$ for an $n = 6$ qubit Quantum Approximate Optimization Algorithm (QAOA) circuit~\cite{farhi2014quantum}, for a path graph with randomly selected values of parameters. This is done for two different isomorphic layouts, each executed four times, using a 27-qubit Fake Backend, which mimics the coupling map and noise profile of a real IBM Quantum device~\cite{fakebackend}. 
Since the layouts are isomorphic, the circuits have the exact same structure, and hence SWAP counts. However, we observe from the plot that the expectation values obtained without any error mitigation or suppression are quite different for the two layouts, with values of $0.481$ 
for layout 1 and $0.383$
for layout 2, and with the ideal expectation being $0.522$. 
The sole reason for this is that the noise of the qubits and 2-qubit connectivity in the two layouts are different. 
Similarly, in Fig.~\ref{fig:why_we_need_best_layout_violin}, we show the wide distribution of the expectation when the same circuit is executed on $104$ isomorphic layouts with different noise profile, and their deviation from the ideal expectation value, i.e., picking a lower quality layout can lead to a high penalty, with expectation values far from the ideal. These two figures show the importance of selecting layouts with low noise profile to achieve high-quality quantum computation.

We next discuss the construction of a QIC for a given user circuit, to obtain the updated noise information with low hardware overheads.

\subsection{Construction of Quality Indicator Circuit (QIC)}

A \textit{Quality Indicator Circuit (QIC)} is defined as any circuit whose ideal noise-free outcome is known \textit{a priori} or can be calculated easily. 
A QIC is targeted, and designed for each layout of a given user circuit.
There are several methods for constructing a QIC, such as using a Clifford circuit which can be efficiently simulated, or a mirrored circuit whose outcome is known. However, creating such a QIC for any given user circuit creates a QIC whose size is the same as that of the original circuit (or even doubles it, for a mirrored circuit). We instead opt for a construction where the depth of the QIC can be reduced relative to the user's circuit, while retaining the structure of the 2-qubit gates in the original circuit. Our QIC construction ensures that: (i) its ideal outcome is known, (ii) it retains the basic template of the original circuit, and (iii) it can often be constructed with a smaller depth than the original circuit.

A QIC has a network of CNOT gates between arbitrary pairs of qubits, which is inspired by the 2-qubit gates in the original circuit, and is sandwiched between two layers of Hadamard gates acting on all qubits. Algorithm~\ref{alg:qic} explains the procedure to construct a QIC for a given circuit. It constructs the QIC by scanning all 2-qubit gates in the user's circuit $q$ and mimics similar CNOT gate structures with a possibility of minimized depth sandwiched between the two layers of Hadamard gates.
The time complexity of such a construction is $\mathcal{O}(n.d)$, where $n$ is the number of qubits and $d$ is the depth of the circuit.

\begin{algorithm}[t]
\caption{QIC construction for a given user circuit}
\label{alg:qic}
\begin{algorithmic}[1]
    \REQUIRE User's quantum circuit $q$
    \ENSURE Quality Indicator Circuit $QIC$
    \STATE $2qPerPair \gets \{~\}$
    \FORALL{2-qubit gate $g = (q_i, q_j) \in q$}
        \STATE $pair \gets$ \textsc{Sort}$(q_i, q_j)$
        \STATE $2qPerPair[pair] \gets 2qPerPair[pair] + 1$
    \ENDFOR
    \STATE $minCount \gets$     
    \textsc{Min}$(2qPerPair.values())$
    \FORALL{$pair \in 2qPerPair$}    
        \STATE $2qPerPair[pair] \gets \lceil\frac{2qPerPair[pair]}{minCount}\rceil$
        
    \ENDFOR
    \STATE Initialize $QIC \gets$ \textsc{QuantumCircuit}$(q.numQubits)$
    \STATE Apply Hadamard gate on all the qubits of \texttt{QIC}
    \FORALL{ $pair \in 2qPerPair$}
            \STATE Apply \textsc{CNOT gate} on $ (pair[0], pair[1])$ to $QIC$, $2qPerPair[pair]$ times     
    \ENDFOR
    \STATE Apply Hadamard gate on all the qubits in $QIC$
\end{algorithmic}
\end{algorithm}

\begin{lemma}
\label{lemma:qic_construction_time}
Given a circuit of depth $d$ with $n$ qubits, Algorithm~\ref{alg:qic} constructs the corresponding QIC in $\mathcal{O}(n.d)$ time.
\end{lemma}

\begin{proof}
    See Appendix~\ref{proof1}.
\end{proof}

In Fig.~\ref{fig:original-cir-and-qic}, we show an example of a QIC circuit construction for a 6-qubit $p = 2$ QAOA circuit for a path graph. The 2-qubit depth, i.e., the depth considering only 2-qubit gates, of this QAOA circuit is $8$, whereas that of the QIC circuit is $2$.

Next, in Lemma~\ref{lemma:qic}, we show that the ideal noise-free outcome of a QIC circuit, constructed using Algorithm~\ref{alg:qic}, is always $\ket{0}^{\otimes n}$, with $n$ being the number of qubits. 

\begin{figure}[t]
    \centering
    \begin{subfigure}[b]{0.7\textwidth}
        \centering
        \includegraphics[width=\textwidth]{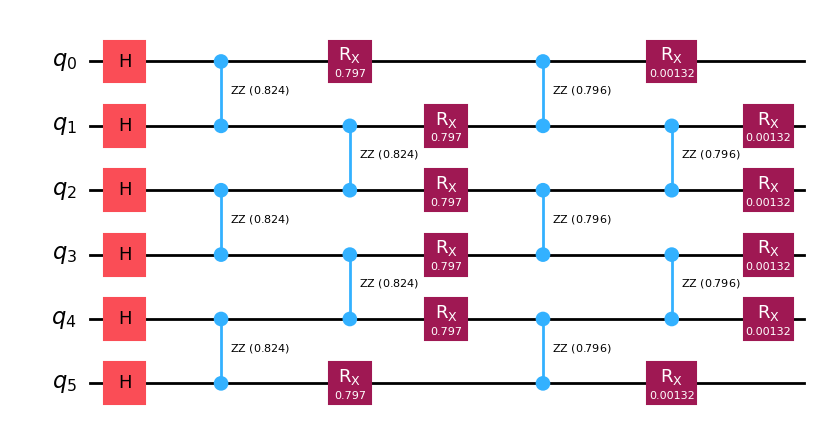}
        \caption{A 6-qubit $p=2$ QAOA \textit{user circuit} for a path graph}
        \label{fig:large-depth-original-circuit}
    \end{subfigure}
    \begin{subfigure}[b]{0.3\textwidth}
        \centering
        \includegraphics[width=\textwidth]{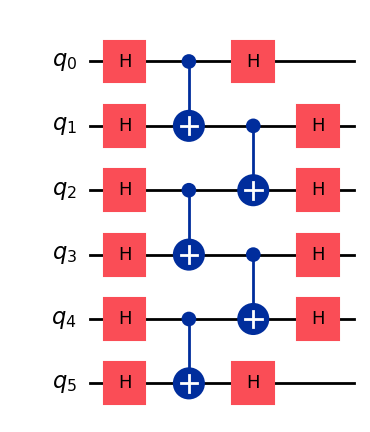}
        \caption{\textit{QIC} corresponding to the user circuit}
        \label{fig:small-depth-QIC}
    \end{subfigure}
    \caption{QIC formulation for a given user circuit. The user circuit in (a) has a 2-qubit depth of eight since each $R_{ZZ}$ gate consists of two CNOT gates. The corresponding QIC in (b) retains the 2-qubit gate structure of the user circuit but has a 2-qubit depth of two.}
    \label{fig:original-cir-and-qic}
\end{figure}

\begin{lemma}
\label{lemma:qic}
    The ideal noise-free outcome of an $n$-qubit quantum circuit with a network of CNOT gates between arbitrary pairs of qubits sandwiched between two layers of Hadamard gates acting on all qubits is $|0\rangle ^{\otimes n}$.
\end{lemma}

\begin{proof}
    See Appendix~\ref{proof2}.
\end{proof}

Since the ideal outcome of an $n$-qubit QIC is always $|0\rangle^{\otimes n}$, the fidelity of the outcome of executing the QIC on a backend with the ideal outcome is an \emph{indicator} of the layout noise profile. Therefore, we can execute the QIC $M$ times (called \textit{shots}) on a particular layout, and observe the number of times $m$, the ideal outcome, is obtained. The ratio $\frac{m}{M}$ is called the \emph{QIC score}, and indicates the noise profile of the layout. The higher the score, the better is the layout.

It is also possible to define an equivalent scoring mechanism using observables. We define any $Z$ type observable as $Z^b$, $b \in \{0,1\}^n$ where $0$ indicates $I$ and $1$ indicates $Z$. For example, $Z^{01001} = IZIIZ$. A corollary of Lemma~\ref{lemma:qic} is that the ideal expectation value of a $Z^b$ observable, for any $b$, for a QIC circuit is $+1$. Therefore, the expectation value of a $Z^b$ observable is also an equivalent QIC score. However, for this type of scoring, the choice of observable can be crucial since higher weight observables are known to be more susceptible to noise. A possible observable can be $\frac{1}{n}\sum_{i=1}^n Z_i$, i.e., the average of all weight-1 observables on $n$ qubits. Such an observable ensures that the contribution of the noise of each qubit in the layout is considered, and the effect of any particularly good or bad qubit is averaged out. Instead, one can also opt to have the same observable (with any Pauli other than $Z$ converted to $Z$, e.g., $IXIIY \rightarrow IZIIZ$) that they are going to measure for the original circuit to have a good estimate of noise for that observable. We have used metric: $\frac{1}{n-1}\sum\limits_{i=1}^{n-1}\langle Z_i Z_{i+1} \rangle$ for $n$ qubits as the \textit{QIC score} throughout the paper.

\subsection{QIC as a Noise Profile Scoring Technique}

We now establish QIC as a valid scoring technique. For this, (i) we select a static noise model, as available in Fake Backends, (ii) transpile the circuit, (iii) obtain isomorphic layouts for the transpiled circuit, and (iv) score each layout using the QIC method. In Fig.~\ref{fig:qic_scoring_method}, we show that layouts with higher QIC scores indeed lead to a better expectation value for the circuit shown in Fig.~\ref{fig:large-depth-original-circuit} for isomorphic layouts where each QIC is executed four times.

\begin{figure}[t]
    \centering    \includegraphics[width=0.6\columnwidth]{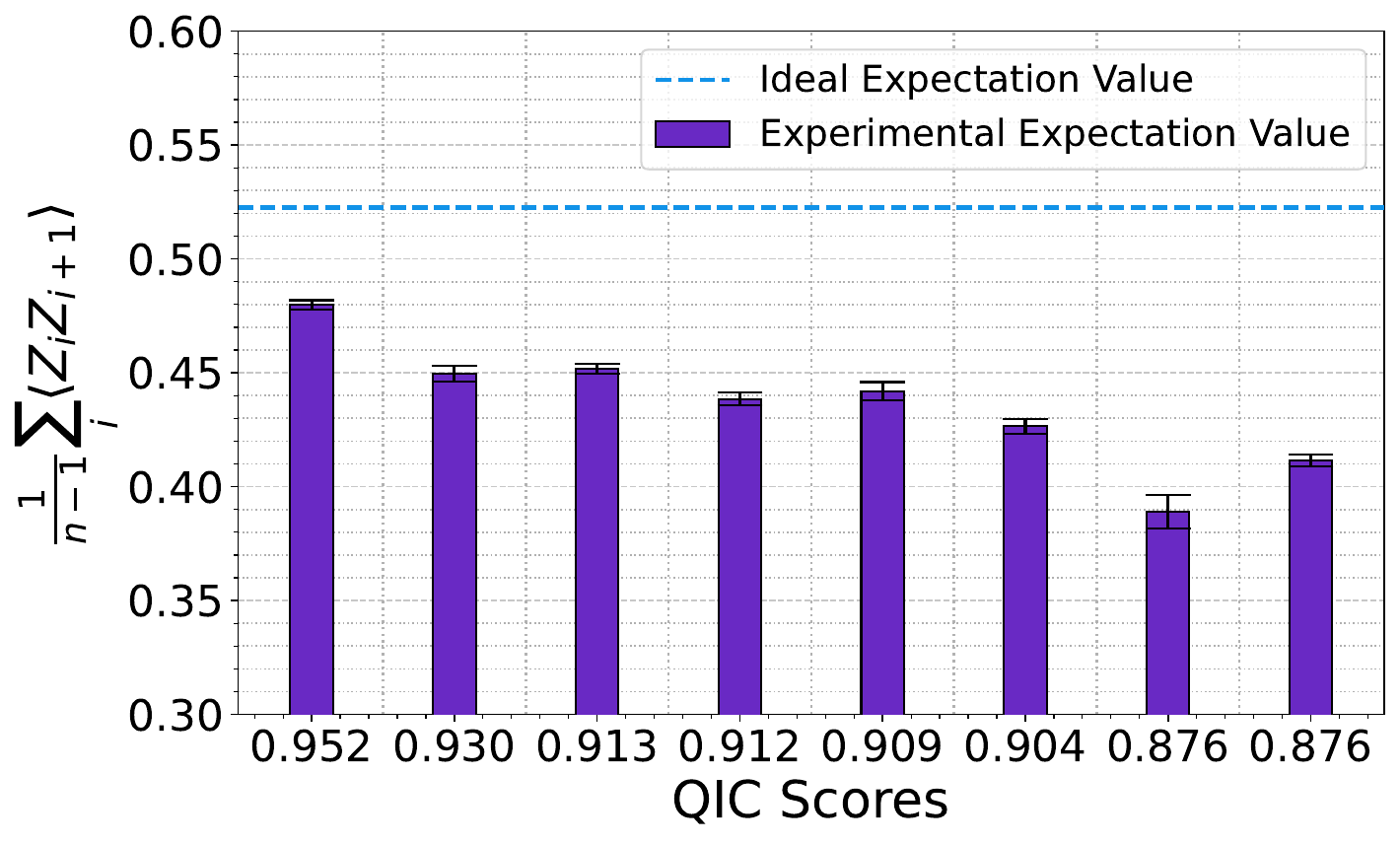}
    \caption{This shows that the expectation value of $\frac{1}{n-1} \sum_i \langle Z_i Z_{i+1} \rangle$ for the QAOA circuit in Fig.~\ref{fig:large-depth-original-circuit} obtained by executing the circuit on layouts with higher QIC score is closer to the ideal value. The quality of the computation drops for layouts with lower QIC scores. This establishes QIC as a technique to score the noise profile of layouts.}
    \label{fig:qic_scoring_method}
\end{figure}

For a static noise model, QIC and Mapomatic are expected to provide comparable scores for the isomorphic layouts. For Mapomatic, the lower the score, the better the layout, and vice versa for QIC. This is confirmed in Fig.~\ref{fig:qic_mm}, where we show that the QIC scores and the Mapomatic scores are inversely comparable for the isomorphic layouts considered in Fig.~\ref{fig:qic_scoring_method}. We see that the QIC score deteriorates with the layouts, while the Mapomatic score increases. Therefore, when there is no deviation in the system noise, Mapomatic and QIC are equivalent techniques for scoring the layouts according to their noise profile.

\begin{figure}[t]
    \centering    \includegraphics[width=0.6\columnwidth]{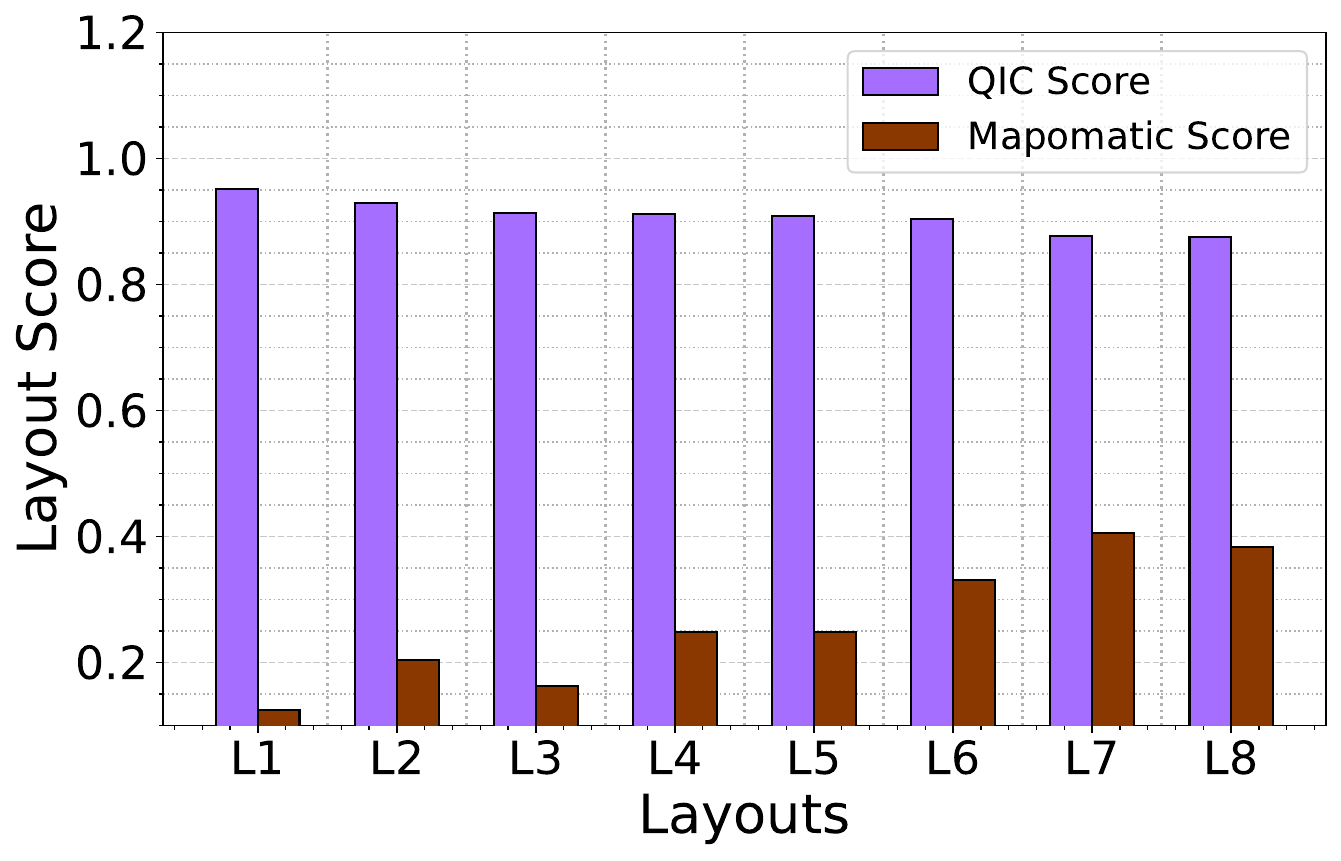}
    \caption{This figure shows the layout score of the eight layouts considered in Fig.~\ref{fig:qic_scoring_method} using both QIC and Mapomatic for a static noise model. The layouts are arranged according to increasing noise profile. We note that the QIC scores deteriorates, and Mapomatic score increases with the layouts, thus verifying that for a static noise model QIC and Mapomatic as equivalent scorers.}
    \label{fig:qic_mm}
\end{figure}

\subsection{Depth Scaling in 6-qubit QAOA Circuit}

For a highly structured circuit, such as the QAOA path circuit of Fig.~\ref{fig:large-depth-original-circuit}, the same layer is repeated multiple times. Such a structure is often seen in circuits of VQE, QAOA, Hamiltonian simulation, QML, etc. For such circuits, the depth of the corresponding QIC remains the same irrespective of the number of times the layer is repeated. One can easily verify from Algorithm~\ref{alg:qic} that even if we opt for higher $p$ for the QAOA circuit of Fig.~\ref{fig:large-depth-original-circuit}, the corresponding QIC will remain the same. This is an advantage because often the noise profile of the layout scales down exponentially with the depth of the circuit, making it difficult to distinguish between the noise profile of two layouts for deep circuits. Since QIC captures only the main structure and the minimum information from the circuit, it can avoid the unnecessary exponential down-scaling of the score with depth.

In Fig.~\ref{fig:depth-scaling-qaoa}, we show that, in general, the QIC score remains a valid and reliable metric irrespective of the depth of the circuit by observing that the \textit{fidelity}~\cite{nielsen2001quantum} of the QAOA circuit increases as the QIC score increases. We have executed each circuit 4 times. We see that $0 \leq \text{fidelity} \leq 1$, where the fidelity indicates the closeness of experimental measurement statistics with ideal measurement statistics, i.e. a fidelity of $1$ means that the experimental outcome matches the ideal outcome.

\begin{figure}[t]
    \centering
    \begin{subfigure}{0.49\textwidth}
        \centering
        \includegraphics[width=\linewidth]{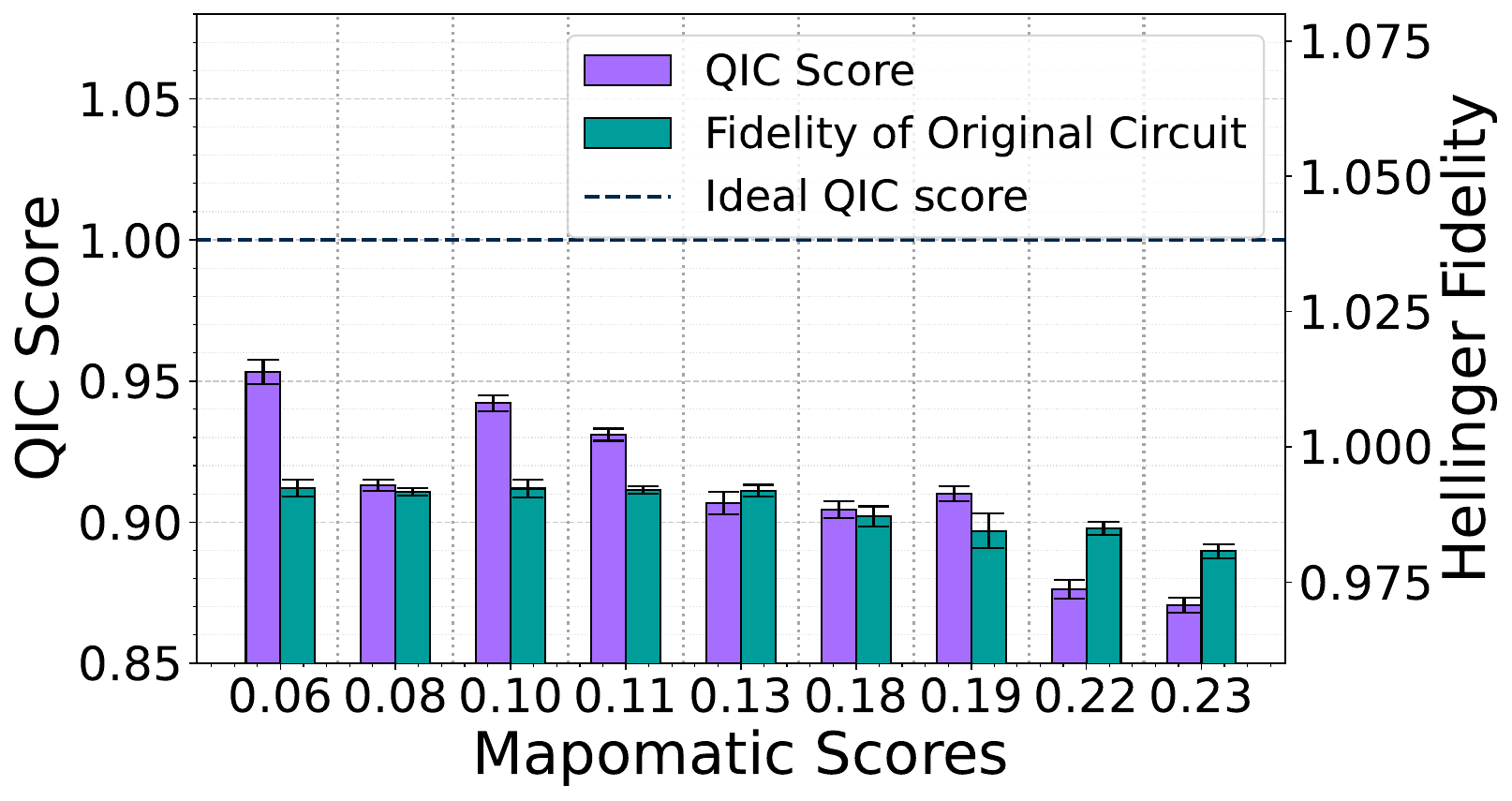}
        \caption{Number of Layers =1}
    \end{subfigure}\qquad
    \begin{subfigure}{0.45\textwidth}
        \centering
        \includegraphics[width=\linewidth]{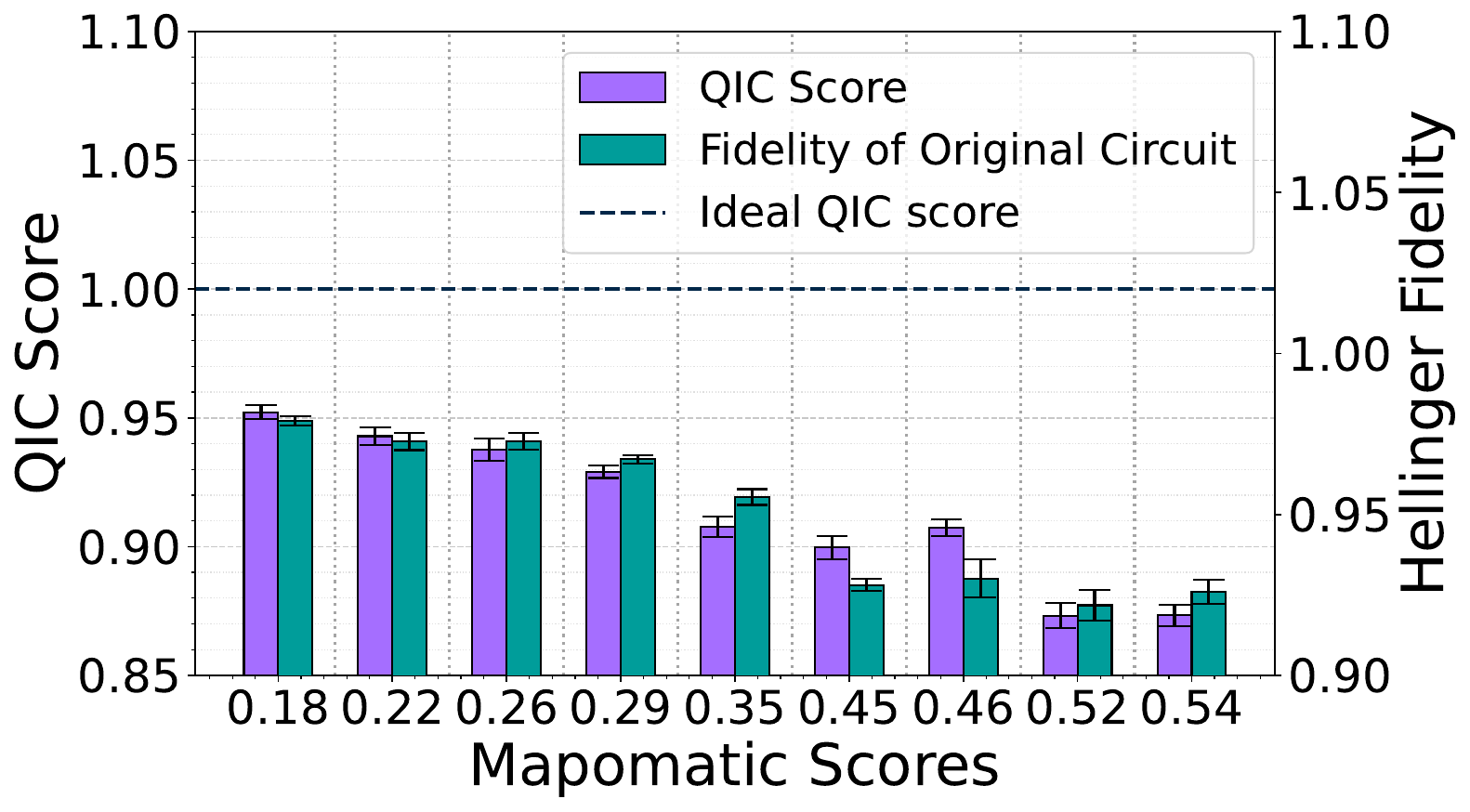}
        \caption{Number of Layers =3}
    \end{subfigure}\\
    \begin{subfigure}{0.45\textwidth}
        \centering  \includegraphics[width=\linewidth]{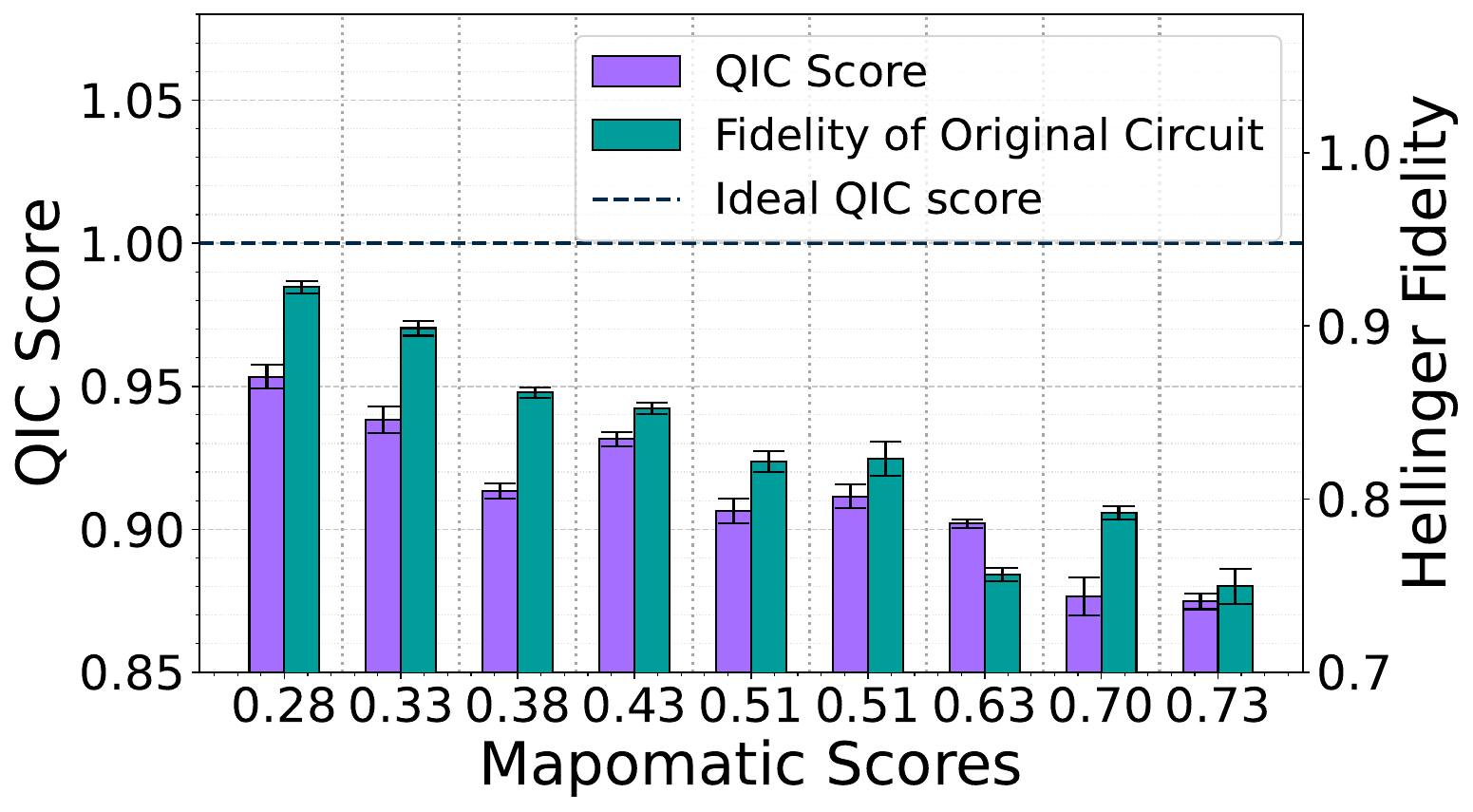}
        \caption{Number of Layers =5}
    \end{subfigure}\qquad
    \begin{subfigure}{0.45\textwidth}
        \centering
        \includegraphics[width=\linewidth]{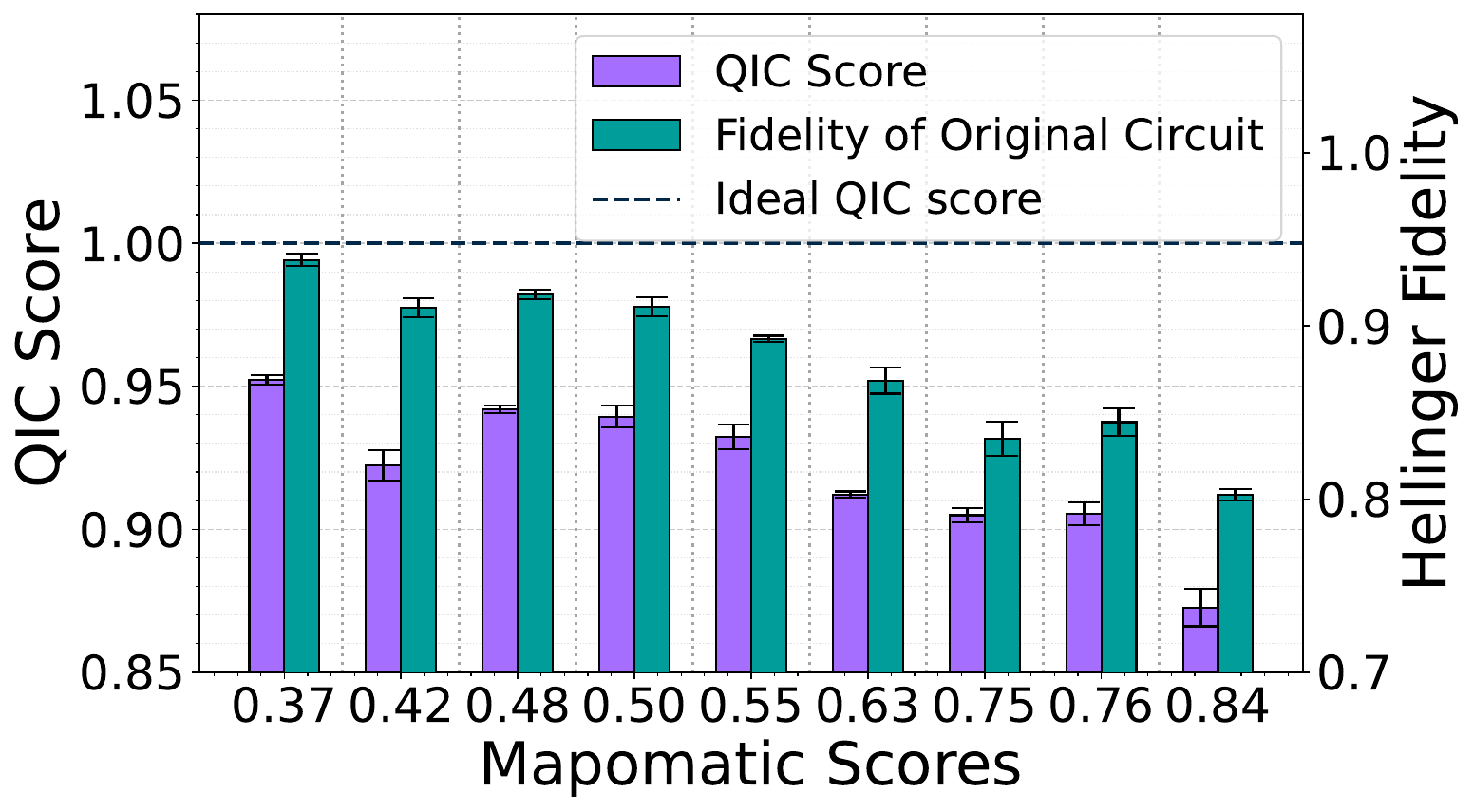}
        \caption{Number of Layers =7}
    \end{subfigure}\\
    \begin{subfigure}{0.49\textwidth}
        \centering        \includegraphics[width=\linewidth]{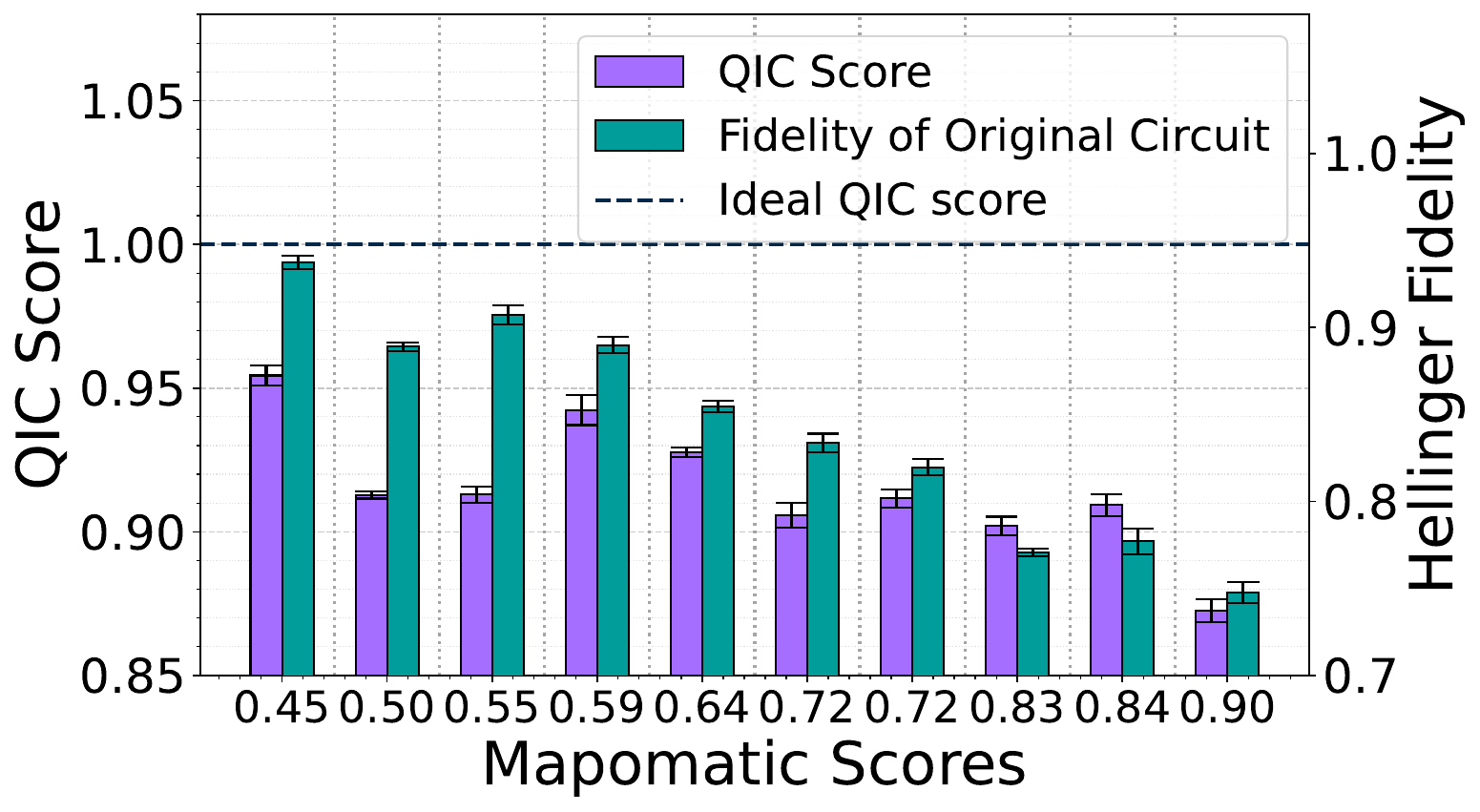}
        \caption{Number of Layers =9}
    \end{subfigure}\qquad
    \begin{subfigure}{0.45\textwidth}
        \centering        \includegraphics[width=\linewidth]{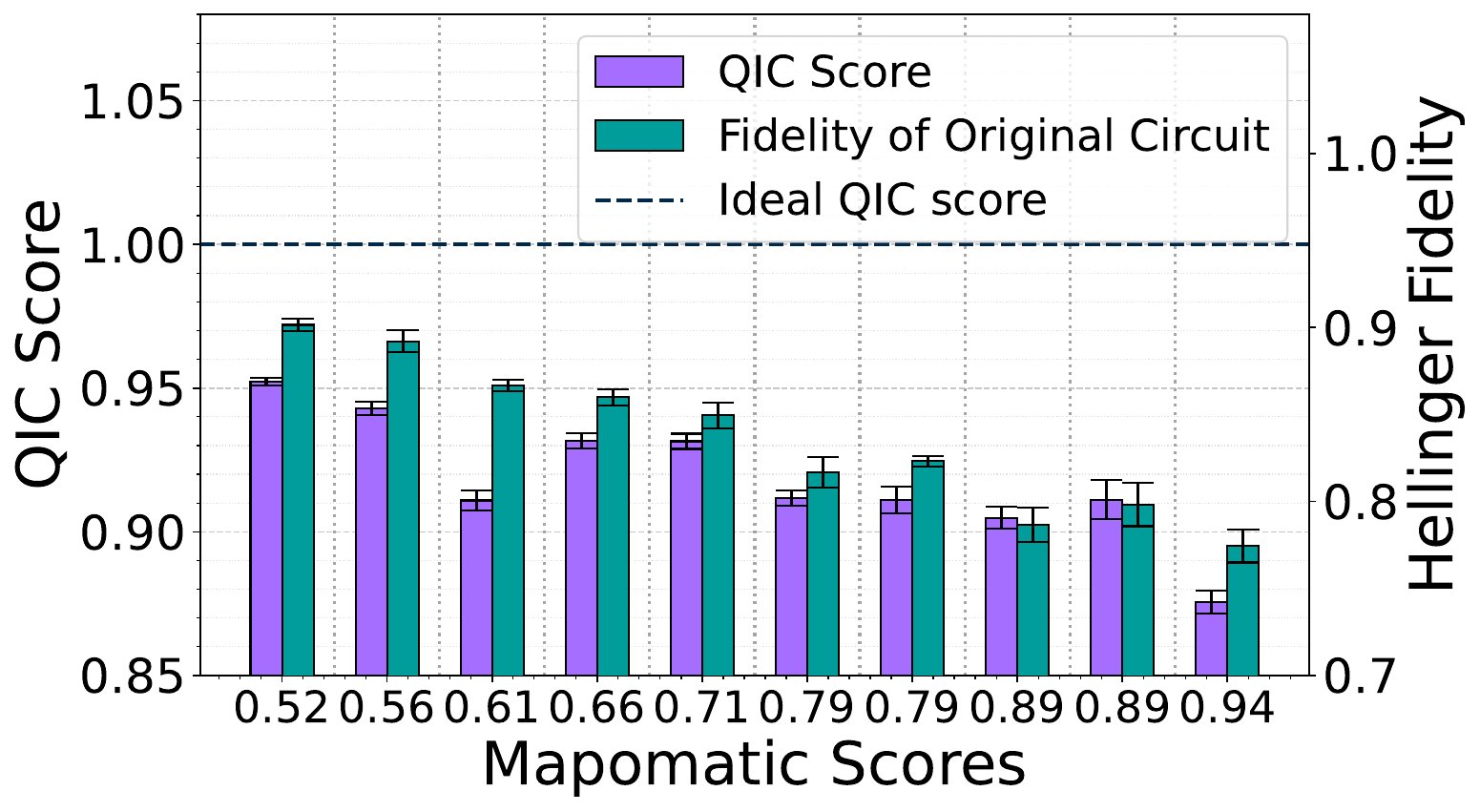}
        \caption{Number of Layers =11}
    \end{subfigure}
    
    \caption{The subfigures correspond to six different values of $p \in \{1, 3, 5, 7, 9, 11\}$ for a 6-qubit QAOA circuit. These experiments are performed on a 27-qubit Fake Backend. We observe that even when the depth of the circuit increases, the QIC score remains a valid and reliable metric to find the best layout.}
    \label{fig:depth-scaling-qaoa}
\end{figure}
\subsection{Resource Usage Comparison between QIC and JIT}
The number of experiments to be performed for JIT remains unchanged irrespective of the size of the circuit since it characterizes the entire hardware. For a 27-qubit IBM Quantum device (FakeKolkataV2), the number of circuit executions required for JIT is $132$. On the other hand, the basic QIC method needs to be executed for each isomorphic layout for a given circuit. To compare the resource requirements of JIT and QIC, we vary the number of qubits in the QAOA circuit and obtain the number of isomorphic layouts. Assuming the same number of shots for each circuit, Table~\ref{tab:shots_JIT_vs_QIC} shows the execution overheads of QIC and JIT.

\begin{table}[t]
\centering
\caption{Resource comparison between JIT, and vanilla QIC for QAOA circuits of different width on FakeKolkataV2}
\label{tab:shots_JIT_vs_QIC}
\small
\begin{tabular}{@{}p{2.6cm}cccccp{}@{}}
\toprule
Circuit executions & \textbf{6 qubits} & \textbf{10 qubits} & \textbf{14 qubits} & \textbf{18 qubits} & \textbf{20 qubits}\\
\midrule
\textbf{JIT} & 132 & 132 & 132 & 132 & 132\\
\textbf{QIC} & 104 & 156 & 128 & 100 & 88 \\
\bottomrule
\end{tabular}
\end{table}

From the table, it is evident that while the basic QIC method, in general, has lower overhead than JIT, there can be scenarios where the this QIC approach requires more circuit executions, e.g., 10-qubit QAOA. This motivates the need to reduce the resource consumption of the basic QIC method. Next, 
we propose a method to unite multiple isomorphic layouts, and execute a single larger QIC for them, which significantly reduces the execution overheads.

\section{Union QIC Approach to Reduce Circuit Execution}
\label{sec:union_disjoint_qic}

In this section we discuss techniques to reduce the circuit execution overhead for QIC. One potential way is to filter out some layouts based on the Mapomatic score. For example, if there are $m$ layouts, one can choose to ignore, say, the last $\frac{m}{2}$ layouts based on their Mapomatic score. This idea makes use of the inherent assumption that a very good Mapomatic layout cannot become very bad over time, even with stale calibration data, and vice versa. However, this heuristic needs to be validated through rigorous data generation over time from the hardware, and may not hold true across different calibration cycles. Therefore, we instead propose a method in which we create the union of multiple QICs, thus reducing the number of circuit executions.

\begin{figure}[t]
    \centering
    \includegraphics[width=0.7\columnwidth]{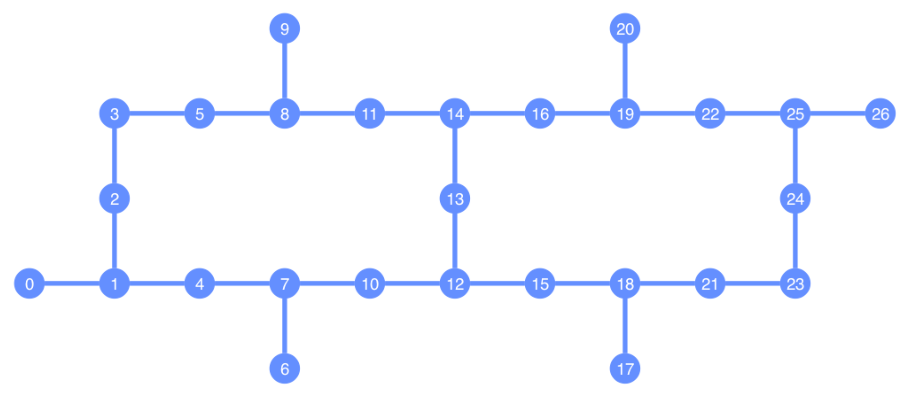}
    \caption{The coupling map of a 27-qubit IBM Quantum device}
    \label{fig:fake_kolkata}
\end{figure}

For a given user circuit, the number of QIC executions is equal to the number of isomorphic layouts. However, na\"{i}vely, one can unite multiple disjoint layouts, execute a single QIC for them, and obtain the scores for each layout by marginalizing over the rest of the qubits~\cite{bhoumik2024resource}. For example, Fig.~\ref{fig:fake_kolkata} shows the coupling map of the 27-qubit FakeKolkataV2. Consider two disjoint isomorphic layouts $l_1 = [18, 21, 23, 24, 25, 26]$ and $l_2 = [3, 2, 1, 4, 7, 6]$ for the 6-qubit QAOA (Fig.~\ref{fig:large-depth-original-circuit}). We can create the QIC (Fig.~\ref{fig:small-depth-QIC}) individually for each of the two layouts, combine them together to form a 12-qubit \emph{union QIC}, and execute the union QIC on the hardware. The outcome for the first QIC corresponding to layout $l_1$ can be obtained by marginalizing over $l_2$ and vice versa. Therefore, instead of two circuit executions, a single execution is sufficient.

However, all the isomorphic layouts cannot be disjoint. Even for this 6-qubit scenario, there exist multiple layouts which overlap with each other (e.g., $[19, 22, 25, 24, 23, 21],\allowbreak\ [18, 21, 23, 24, 25, 22],\allowbreak\ [19, 16, 14, 13, 12, 15],\allowbreak\ [20, 19, 22, 25, 24, 23]$, etc.), and cannot be combined as the previous example. In other words, out of all the isomorphic layouts, there exist sets of mutually disjoint layouts that can be combined and executed as a single union QIC circuit. So, the total number of QIC circuit executions reduces to the same as the number of such sets.  
This motivates the problem of finding
the \textit{minimal} number of such disjoint sets, and this problem 
can be stated as follows:

\textbf{Problem.~}
\label{problem}
Given a circuit $C$ with $N$ isomorphic layouts for a given hardware, find the minimum number of disjoint sets such that $S_i \cap S_j = \phi, \forall \text{~sets~} S_i \neq S_j \text{~and~} \text{~layouts~} l^p, l^r \in S_i, l^p \cap l^r = \phi , \forall i$.
\noindent In other words, we want to partition the layouts into the minimum number of sets such that the layouts in each set are disjoint. Then we can execute a single QIC for each disjoint set.

\begin{theorem}
\label{thm:chromatic}
    Finding the minimum number of disjoint sets $S = \{S_1, S_2, ..., S_k\}$  
    is NP-Hard.
\end{theorem}

\begin{proof}
    See Appendix~\ref{proof3}.
\end{proof}

Since the problem is NP-Hard, we design a \textit{greedy heuristic} to solve it. A layout $l$ can be assigned to a set $S$ only if $l \cap l^s = \phi$ $\forall$ $l^s \in S$. If a layout cannot be assigned to any existing set, a new set is created and $l$ is assigned to it. For every isomorphic layout $l$, Algorithm~\ref{alg:disjointness} checks whether it is compatible (disjoint) with every existing layout in a given set. Algorithm~\ref{alg:greedy} starts with one empty set containing the first layout and uses Algorithm~\ref{alg:disjointness} to assign all $N - 1$ isomorphic layouts one by one to some existing compatible set, or if not possible, creates a new empty set and assign $l$ to it. In this greedy approach, a new set is created only if the layout can not be assigned to any existing sets.

\begin{algorithm}[t]
    \caption{Algorithm for \textsc{IsLayoutDisjointWithSet}}
    \label{alg:disjointness}
    \begin{algorithmic}[1]
        \REQUIRE A set of mutually disjoint layouts $S$, a different layout $l$
        \ENSURE \textbf{True} if $l \cap l^s = \phi$ $\forall$ $l^s \in S$, else \textbf{False}
        \STATE $compatible \leftarrow \textbf{True}$
        \FORALL{$l^s \in S$}
            \STATE $overlap \leftarrow 0$
            \FORALL{$l_i \in l$}
                \IF{$l_i \in l^s$}
                    \STATE $overlap = overlap +1$
                \ENDIF
            \ENDFOR
            \IF{$overlap > 0$}
                \STATE $compatible \leftarrow \textbf{False}$
                \STATE break 
            \ENDIF
        \ENDFOR
        \RETURN $compatible$
    \end{algorithmic}
\end{algorithm}

\begin{algorithm}[t]
\caption{Greedy construction of a collection of sets containing disjoint layouts.}
\label{alg:greedy}
\begin{algorithmic}[1]
\REQUIRE A list $L$ of isomorphic layouts for a given circuit and hardware
\ENSURE A list $U = \{S_1, S_2, ..., S_k\}$ where each $S_i, 1 \leq i \leq k$ contains disjoint layouts
\STATE $U \leftarrow [~]$
\STATE $L_R \leftarrow$ a random permutation of $L$
\STATE Add set $\{L_R[0]\}$ to $U$
\FORALL{layout $l \in L_R[1:]$}
    \STATE $placed \gets \textbf{False}$
    \FORALL{set $S_i \in U$}
        \STATE $compatible \gets$ \textsc{IsLayoutDisjointWithSet}$(S_i, l)$
        \IF{$compatible$} 
            \STATE $S_i = S_i \cup l$
            \STATE $placed \gets \textbf{True}$
            \STATE \textbf{break}
        \ENDIF
    \ENDFOR
    \IF{not $placed$}
        \STATE Append new set $\{l_C\}$ to $U$
    \ENDIF
\ENDFOR 
\RETURN $U$
\end{algorithmic}
\end{algorithm}

\begin{lemma}
\label{alg2}
    Algorithm~\ref{alg:disjointness} checks the compatibility of a layout $l$ with a set of mutually disjoint layouts $S = \{l^1, l^2, \hdots, l^m\}$ with $\mathcal{O}(n^2.m)$ time complexity, where $n$ is the number of qubits in each layout.
\end{lemma}

\begin{proof}
    See Appendix~\ref{proof4}.
\end{proof}

\begin{lemma}
\label{alg3}
    Algorithm~\ref{alg:greedy} partitions the list of all isomorphic layouts into a minimal number of sets $U = \{S_1, S_2, \hdots, S_k\}$ with $\mathcal{O}(M^3n^2)$ time complexity, where $M$ is the total number of isomorphic layouts and $n$ is the number of qubits in each layout.
\end{lemma}

\begin{proof}
    See Appendix ~\ref{proof5}.
\end{proof}

We obtain the union list $U = \{S_1, S_2, ..., S_k\}$ from Algorithm~\ref{alg:greedy}. In this union list, each set $S_i$ consists of mutually disjoint layouts. Therefore, a single QIC can be constructed by combining the individual QICs corresponding to each layout. This single QIC is executed on the quantum computer, and the output of each individual QIC is obtained by marginalizing over the others.

Algorithm~\ref{alg:greedy} always selects the first layout from the input list of isomorphic layouts $L$ as the first element of the first set. Therefore, the ordering of the input list $L$ can affect the performance of the algorithm. In fact, we observe some deviation in the number of sets created by Algorithm~\ref{alg:greedy} for different orderings of the input list. This deviation is minimal, though, and often by just 1. However, the deviation becomes more prominent when we allow some overlap between the layouts in a set, discussed in Sec.~\ref{sec:distortion}. Therefore, we opt for a random permutation of the input list of isomorphic layouts $L$, and run Algorithm~\ref{alg:greedy} in parallel for various permutations. The parallel execution can be performed using \emph{Qiskit Serverless}~\cite{gambetta2024quantum} or in any other HPC environment.

\begin{table}[t]
\centering
\caption{Resource comparison between JIT, basic QIC, and Union QIC method for QAOA circuits of different widths on FakeKolkataV2}
\label{tab:QIC_vs_union_QIC}
\small
\begin{tabular}{lccccc}
\toprule
\bf Circuit executions & \textbf{6 qubits} & \textbf{10 qubits} & \textbf{14 qubits} & \textbf{18 qubits} & \textbf{20 qubits}\\
\midrule
\textbf{JIT} & 132 & 132 & 132 & 132 & 132\\
\textbf{Basic QIC} & 104 & 156 & 128 & 100 & 88 \\
\textbf{Union QIC} & 54 & 120 & 128 & 100 & 88\\
\bottomrule
\end{tabular}
\end{table}

Table~\ref{tab:QIC_vs_union_QIC} shows the number of circuit executions required when union QIC is created for disjoint layouts. We note that while it reduces the number of circuit executions over basic QIC for user circuits with fewer number of qubits, there is no savings seen for larger circuits. This is expected since the bigger the circuit, the more difficult it is to find non-overlapping isomorphic layouts. In the following section we relax this requirement to allow some overlaps between layouts in an attempt to reduce the number of circuit executions even further.

\section{Overlap QIC with Distortion Thresholds}
\label{sec:distortion}

In Union QIC,  
not many isomorphic layouts for larger circuits conform to the requirement of disjointness. When this is possible, the reduction in the number of QIC executions was $2\times$, implying that each set $S_i$ in the union list $U$ contained only two disjoint layouts on average. In Fig.~\ref{fig:ceil_avg}, we show three QICs for three \emph{overlapping} isomorphic layouts. Algorithm~\ref{alg:greedy} will place these layouts in three different sets to maintain disjointness between the layouts inside every set.

\begin{figure*}[t]
    \centering
    \begin{subfigure}[b]{0.7\textwidth}
        \includegraphics[width=\textwidth]{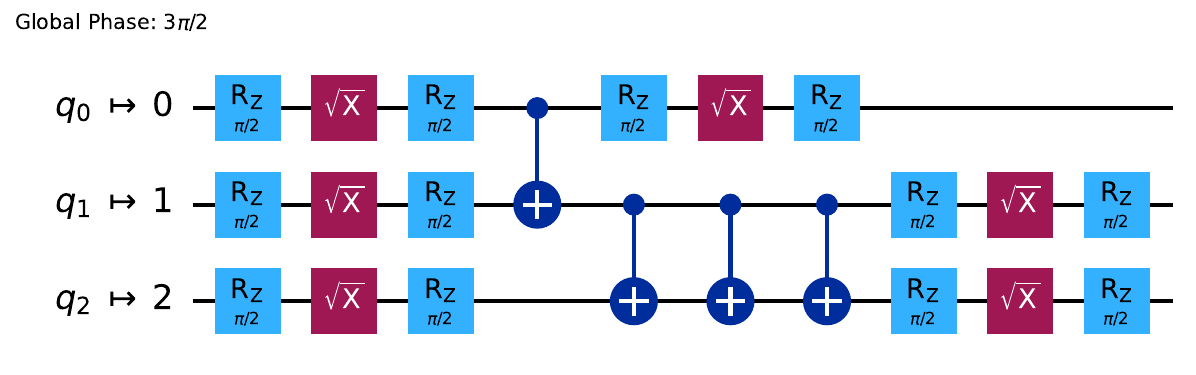}
        \caption{Layout = [0,1,2]}
        \label{fig:ceil_avg_a}
    \end{subfigure}\\
    \begin{subfigure}[b]{0.7\textwidth}
        \includegraphics[width=\textwidth]{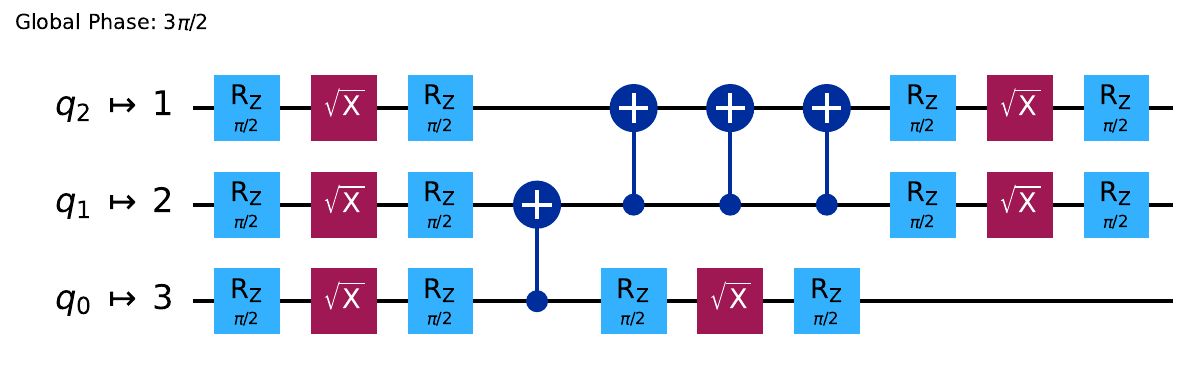}
        \caption{Layout = [3,2,1]}
        \label{fig:ceil_avg_b}
    \end{subfigure}\\
    \begin{subfigure}[b]{0.7\textwidth}
        \includegraphics[width=\textwidth]{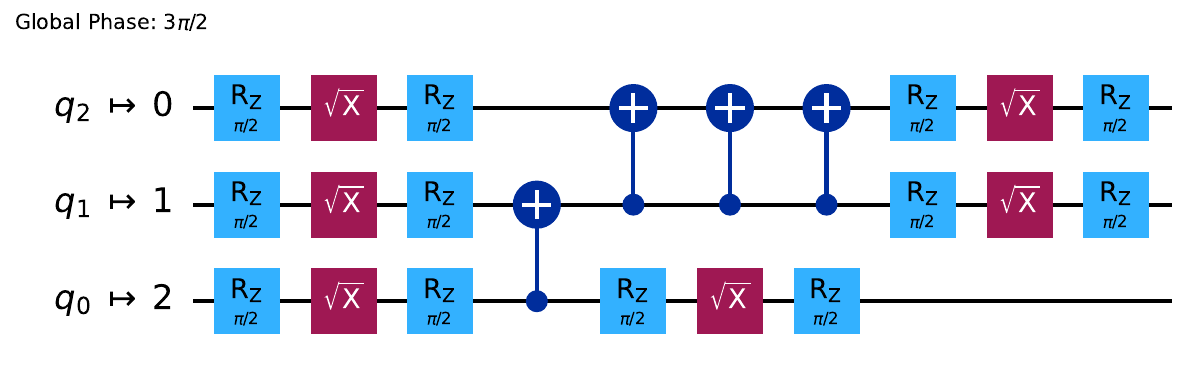}
        \caption{Layout = [2,1,0]}
        \label{fig:ceil_avg_c}
    \end{subfigure}
    \caption{Three example QICs mapped on three disjoint isomorphic layouts of a 27-qubit IBM FakeBackend}
    \label{fig:ceil_avg}
\end{figure*}

However, we can reduce the number of circuit executions to just 1 instead of 3 if we allow overlap between the layouts in each set $S_i$ where we, still, can construct a union QIC for these three layout by considering the average number of CNOT gates for each qubit pair. For example, the qubit pair $(0,1)$ contains 1 CNOT gate in the first QIC, and 3 CNOT gates in the third. Hence, the union QIC will contain $(3 + 1)/2 = 2$ CNOT gates in the qubit pair $(0,1)$. On the other hand, qubit pair $(1,2)$ contains 3 CNOT gates in the first and second QIC, and only one CNOT gate in the third QIC. Therefore, the union QIC will contain $\lceil (3 + 3 + 1)/3 \rceil = 3$ CNOT gates. Recall from Sec.~\ref{sec:qic} that the directionality of the 2-qubit gates is not considered during QIC construction. Fig.~\ref{fig:union_qic_ceil_avg} shows the union QIC for the three QICs of Fig.~\ref{fig:ceil_avg} when relaxing the disjointness constraint and permitting overlaps. 

\begin{figure}[t]
    \centering
    \includegraphics[width=0.7\columnwidth]{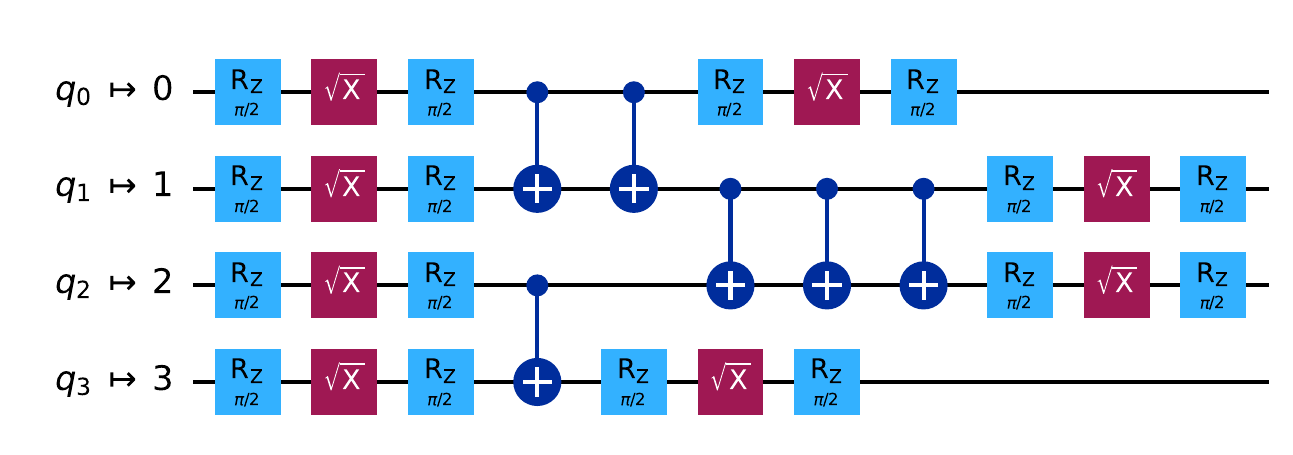}
    \caption{Union QIC allowing layout overlaps, for the QICs showed in Fig.~\ref{fig:ceil_avg}. The number of CNOT gates between each qubit pair is obtained as the average of the CNOT Gates for that qubit pair over all the three layouts.}
    \label{fig:union_qic_ceil_avg}
\end{figure}

Naturally, allowing overlaps changes the CNOT structure from the original QICs to the union QIC. A significant change in the structure will result in incorrect evaluation of the quality of layout when the union QIC is executed and the outcome of the original QIC is obtained via marginalization. We quantify this deviation in the gate structure from the original QIC to the union QIC using the \emph{distortion} metric.

We define \emph{distortion} of a physical qubit as the increase or decrease in the number of 2-qubit gates associated with that physical qubit in the union QIC compared to the original QIC. The \textit{total distortion} of a layout is defined as the sum of the distortion of individual physical qubits in that layout. For example, the distortion of qubit $0$ for layout 1 in Fig.~\ref{fig:ceil_avg_a} is $1$, while its overall distortion is $1 + 1 + 1 = 3$. Taking the sum of individual distortions gives an upper bound on the increase in the 2-qubit depth of the circuit, but captures the deviation faced by each qubit individually.

When constructing union QIC with overlaps (Overlap QIC), we define a \emph{distortion threshold} that limits the acceptable distortion for each layout in order to prevent a significant deviation in the structure of union QIC from the individual ones. We denote that the layout $l^i$ is compatible with the host union $QIC$ within a distortion threshold $T$ as $l^i \cap_T UnionQIC = \phi$.

\begin{algorithm}[t]
\caption{Construct union QIC with distortion, \textsc{UnionWithDistortion}}
\label{alg:construct_distortion}
\begin{algorithmic}[1]
    \REQUIRE A set of layouts $S$ and the individual QIC $Q_l$ corresponding to each $l \in S$
    \ENSURE Union QIC $U_Q$ over all the individual QICs
    \STATE $N = \cup_l\text{num\_qubits}(Q_l)$ \tcc{the total number of qubits in the union QIC is the number of qubits in the union of the physical qubits involved in individual QICs}
    \STATE Construct an empty quantum circuit $U_Q$ with $N$ qubits
    \STATE Apply Hadamard gate on all the qubits of $U_Q$
    \STATE $P = \cup_{Q_l} \text{qubit pairs in }Q_l$ \tcc{union of all qubit pairs in the individual QICs}
    \FORALL{qubit pair $(q_i, q_j) \in P$}
        \STATE $n_l \gets$ number of 2-qubit gates on $(q_i, q_j) \in Q_l$
        \STATE $n_{avg} \gets \lceil \frac{\sum_l n_l}{N_{\hat{l}}} \rceil$ where $N_{\hat{l}}$ is the number of layouts whose QICs have $n_l > 0$.
        \STATE Apply $n_{avg}$ CNOT gates on qubit pair $(q_i, q_j)$ of $U_Q$
    \ENDFOR
    \STATE Apply Hadamard gate on all the qubits of $U_Q$
    \RETURN $U_Q$
\end{algorithmic}
\end{algorithm}

\begin{algorithm}[t]
\caption{Algorithm for \textsc{IsDistortionDisjoint} to find if Layout is Distortion Disjoint with Set}
\label{alg:distortion_compatibility_check}
\begin{algorithmic}[1]
    \REQUIRE A set of layouts $S$ where each layout are compatible with each other as per distortion threshold; individual QIC $Q_{l^S}$ for each layouts $l^S \in S$; a different layout $l$; QIC $Q_l$ corresponding to layout $l$; distortion threshold $T$
    \ENSURE \textbf{True} if $l^i \cap_T Union QIC = \phi$ else \textbf{False}
    \STATE $compatible \gets \textbf{True}$
    \STATE $U_Q =$ \textsc{UnionWithDistortion}$(S \cup l$, $\{Q_{l'}\}$ for $l' \in S \cup l)$
    \FORALL{layout $l' \in S \cup l$}
        \STATE $totalDistortion_{l'} \gets 0$
        \FORALL{qubit $q \in l'$}
            \STATE $count_{q_{l'}} \gets$ number of 2-qubit gates associated with qubit $q \in Q_{l'}$
            \STATE $count_{q_{U_Q}} \gets$ number of 2-qubit gates associated with qubit $q \in U_Q$
            \STATE $totalDistortion_{l'} = totalDistortion_{l'} + | count_{q_{l'}} - count_{q_{U_Q}}|$
        \ENDFOR
        \IF{$totalDistortion_{l'} > T$}
            \STATE $compatible \gets \textbf{False}$
            \STATE \textbf{break}
        \ENDIF
    \ENDFOR
    \RETURN $compatible$
\end{algorithmic}
\end{algorithm}

\begin{table}[t]
    \centering
    \setlength{\tabcolsep}{3pt}
    \def\thickhline{\noalign{\hrule height1pt}}
    \caption{Number of circuit executions required when union QIC is created with distortion threshold 0, 1 and 2}
    \label{tab:permutation}
\small
    \begin{tabular}{c||c|c|c|c|c|c|c|c|c|c|c|c|c|c|c|c|c|c|c|c}
    \thickhline
    \bf Threshold & \multicolumn{20}{c}{\bf \# permutation}\\
    \hline    
     & 1 & 2 & 3 & 4 & 5 & 6 & 7 & 8 & 19 & 10 & 11 & 12 & 13 & 14 & 15 & 16 & 17 & 18 & 19 & 20\\
    \hline\hline
    \bf 0 & 27 & 27 & 27 & 27 & 27 & 27 & 27 & 27 & 27 & 27 & 27 & 27 & 27 & 27 & 27 & 27 & 28 & 28 & 27 & 27\\
    \hline
    \bf 1 & 16 & 16 & 16 & 17 & 17 & 17 & 16 & 17 & 16 & 17 & 16 & 16 & 15 & 18 & 17 & 16 & 16 & 17 & 17 & 18\\
    \hline
    \bf 2 & 12 & 7 & 11 & 8 & 7 & 11 & 11 & 8 & 6 & 11 & 10 & 11 & 6 & 7 & 12 & 11 & 10 & 12 & 11 & 15\\
    \thickhline
    \end{tabular}
\end{table}

\begin{table}[t]
\centering
\caption{Resource Comparison of JIT, Basic QIC, Union QIC and Overlap QIC with a distortion threshold 1, for QAOA Circuits of varying widths on FakeKolkataV2}
\label{tab:distortion_QIC}
\small
\begin{tabular}{lrrrrr}
\toprule
\bf Circuit executions & \textbf{6 qubits} & \textbf{10 qubits} & \textbf{14 qubits} & \textbf{18 qubits} & \textbf{20 qubits}\\
\midrule
\textbf{JIT} & 132 & 132 & 132 & 132 & 132\\
\textbf{Basic QIC} & 104 & 156 & 128 & 100 & 88 \\
\textbf{Union QIC} & 54 & 120 & 128 & 100 & 88\\
\textbf{Overlap QIC} \textit{(Threshold 1)} & 15 & 36 & 33 & 26 & 24\\
\bottomrule
\end{tabular}
\end{table}

\begin{figure}[t]
    \centering        \includegraphics[width=0.5\columnwidth]{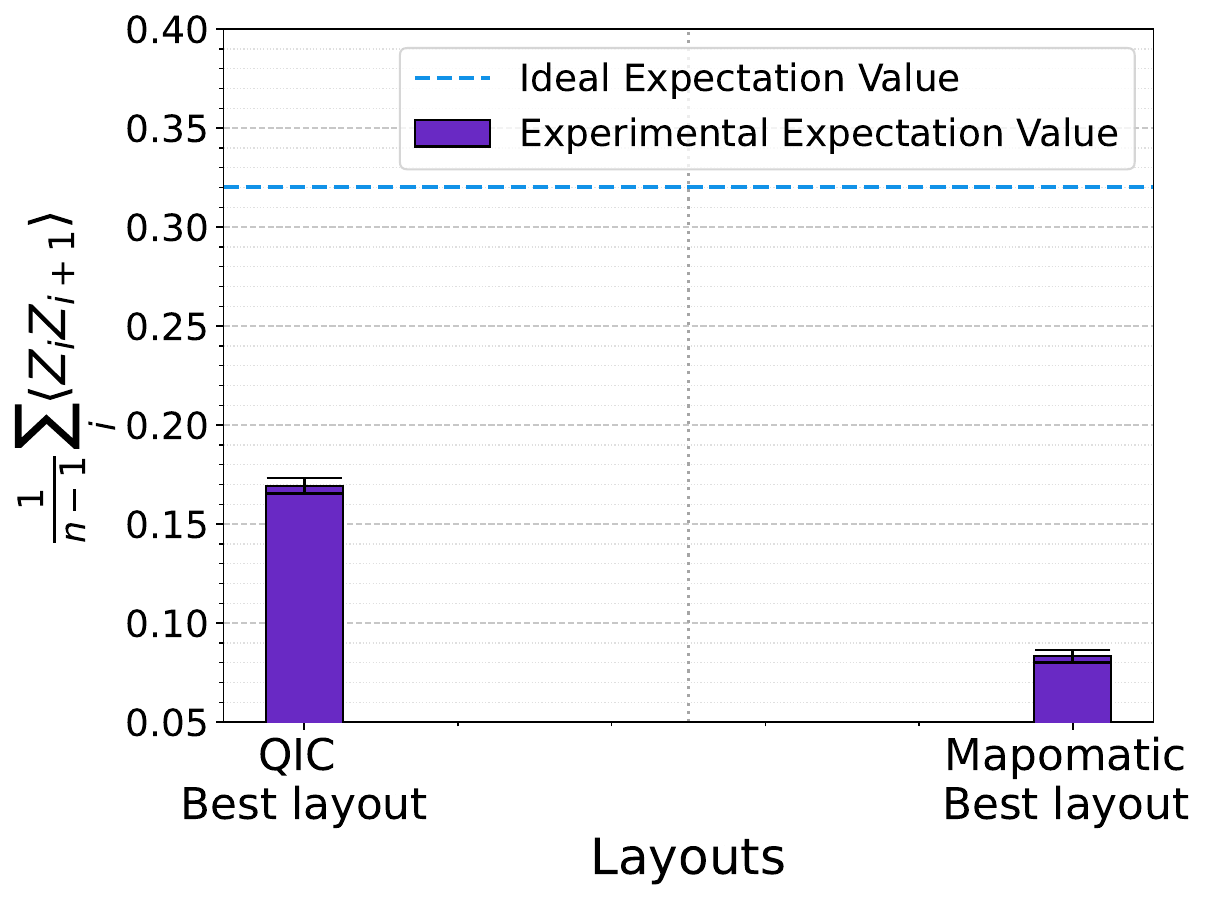}  
    \caption{Expectation values obtained from executing the circuit on the best layout from QIC and Mapomatic on a 127-qubit IBM Quantum hardware. QIC selects a better layout than Mapomatic, seen from the better expectation value obtained.}
    \label{fig: QIC_union_used_in_RH}
\end{figure}

\begin{figure*}[t]
    \centering
    \begin{subfigure}[b]{0.48\textwidth}
        \includegraphics[width=\textwidth]{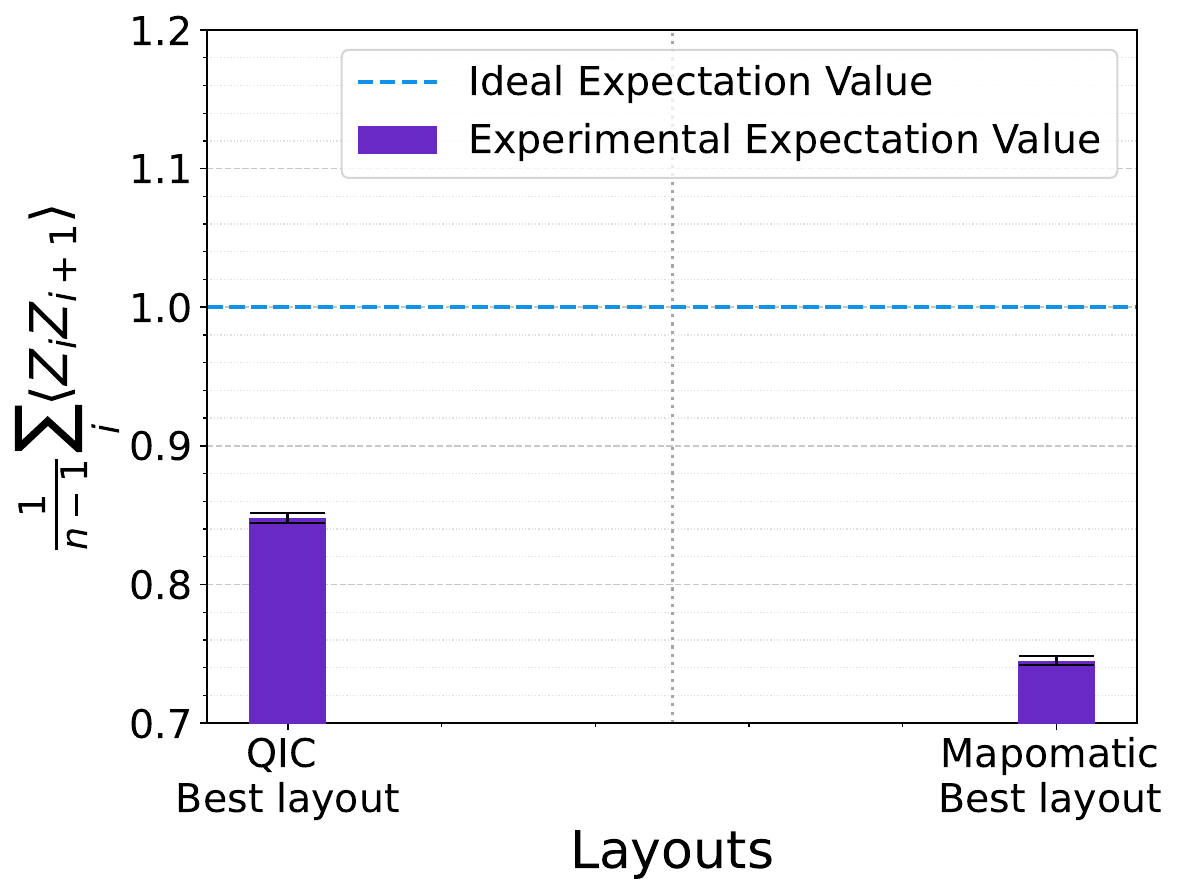}
        \caption{\# Qubits = 6, depth = 8}
        \label{fig:sub1}
    \end{subfigure}~
    \begin{subfigure}[b]{0.48\textwidth}
        \includegraphics[width=\textwidth]{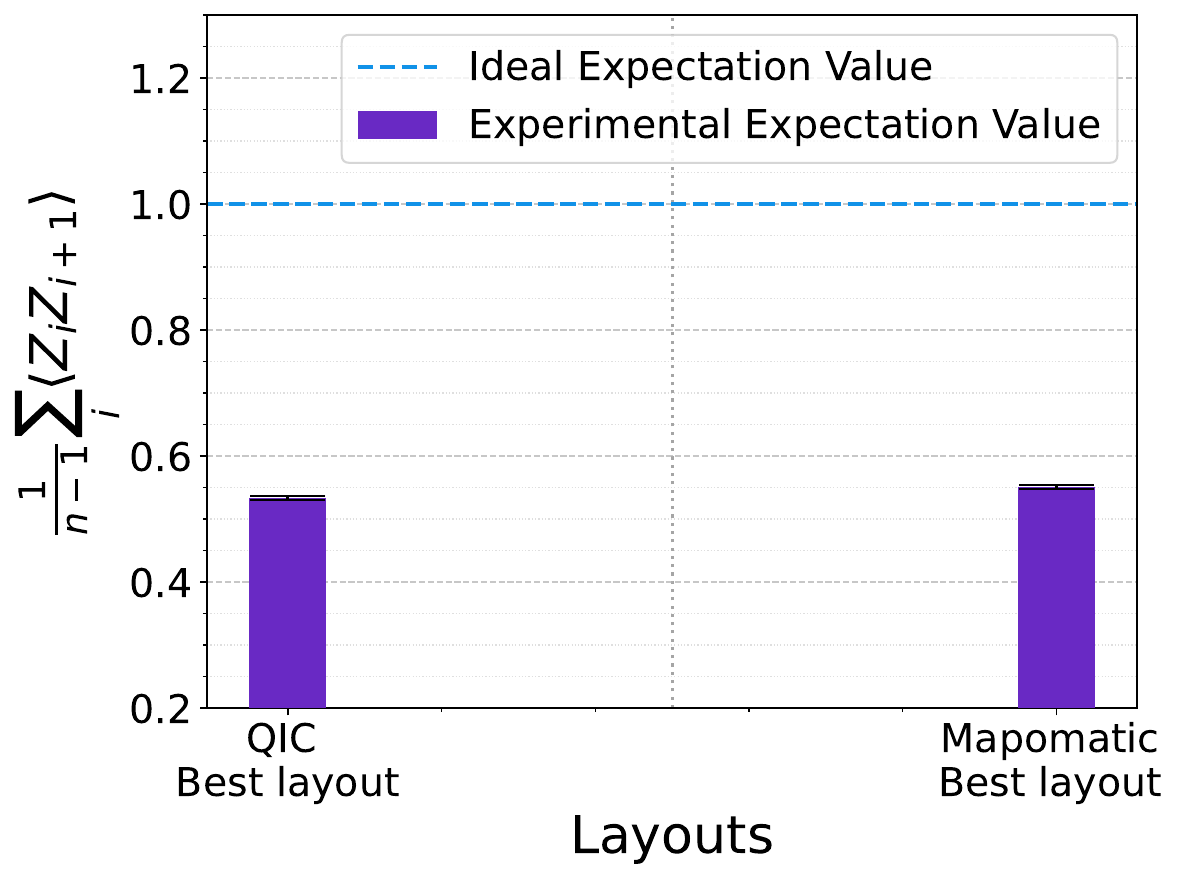}
        \caption{\# Qubits = 12, depth = 8}
        \label{fig:sub2}
    \end{subfigure}\\~\\
    \begin{subfigure}[b]{0.48\textwidth}
        \includegraphics[width=\textwidth]{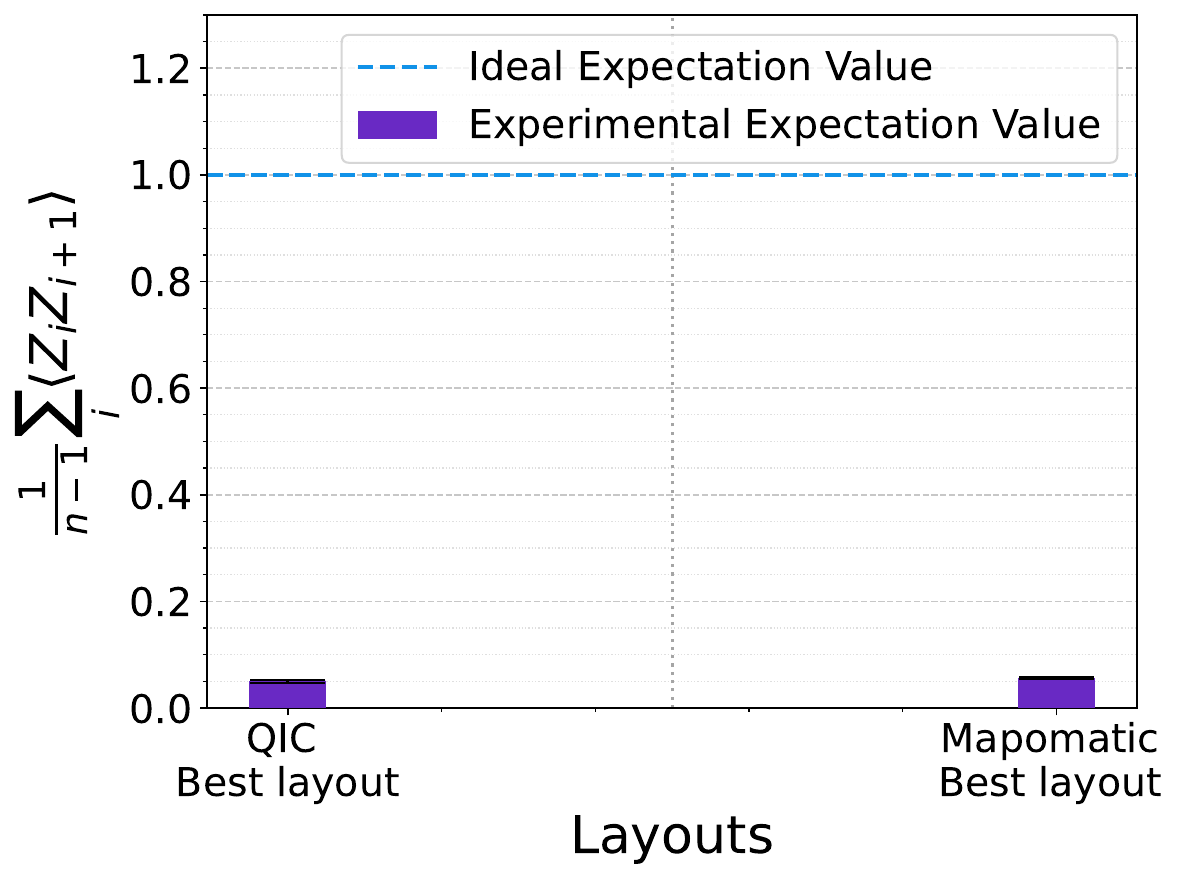}
        \caption{\# Qubits = 16, depth = 8}
        \label{fig:sub3}
    \end{subfigure}
    \caption{Qubit scaling behavior of the QIC method in Identity circuit, obtained using unitary overlap of QAOA circuit. As the qubit count increases, the  best layouts from Mapomatic and QIC became identical.
    }
    \label{fig:qubit_scaling_qic}
\end{figure*}

\begin{algorithm}[t]
\caption{Greedy construction of distortion-compatible sets with overlapping layouts}
\label{alg: greedy_algorithm_to_Reduce_qic_Execs_distortion}
\begin{algorithmic}[1]
    \REQUIRE A list $L$ of isomorphic layouts for a given circuit and hardware; the individual QIC $Q_l$ for each layout $l \in L$; distortion threshold $T$
    \ENSURE A list $U = \{S_1, S_2, \hdots, S_k\}$ where each set $S_i$, $1 \leq i \leq k$ contains layouts which are either disjoint or overlap maintaining the distortion threshold $T$
    \STATE $U \gets [~]$
    \STATE $L_R \gets$ random permutation of $L$
    \STATE Add set $\{L_R[0]\}$ to $U$
    \FORALL{layout $l$ in $L_R[1:]$}
    \STATE $placed \gets \textbf{False}$
    \FORALL{set $S_i \in U$}
        \STATE $compatible \gets$ \textsc{IsDisjoint}($S_i$, $\{Q_{l^s}\}$ $\forall$ $l^s \in S_i$, $l$, $Q_l$, $T$)
        \IF{$compatible$} 
            \STATE $S_i = S_i \cup l$
            \STATE $placed \gets \textbf{True}$
            \STATE \textbf{break}
        \ENDIF
    \ENDFOR
    \IF{not $placed$}
        \STATE Append new set $\{l_C\}$ to $U$
    \ENDIF
    \ENDFOR 
    \RETURN $U$
\end{algorithmic}
\end{algorithm}

\begin{figure}[t]
    \centering        \includegraphics[width=0.5\columnwidth]{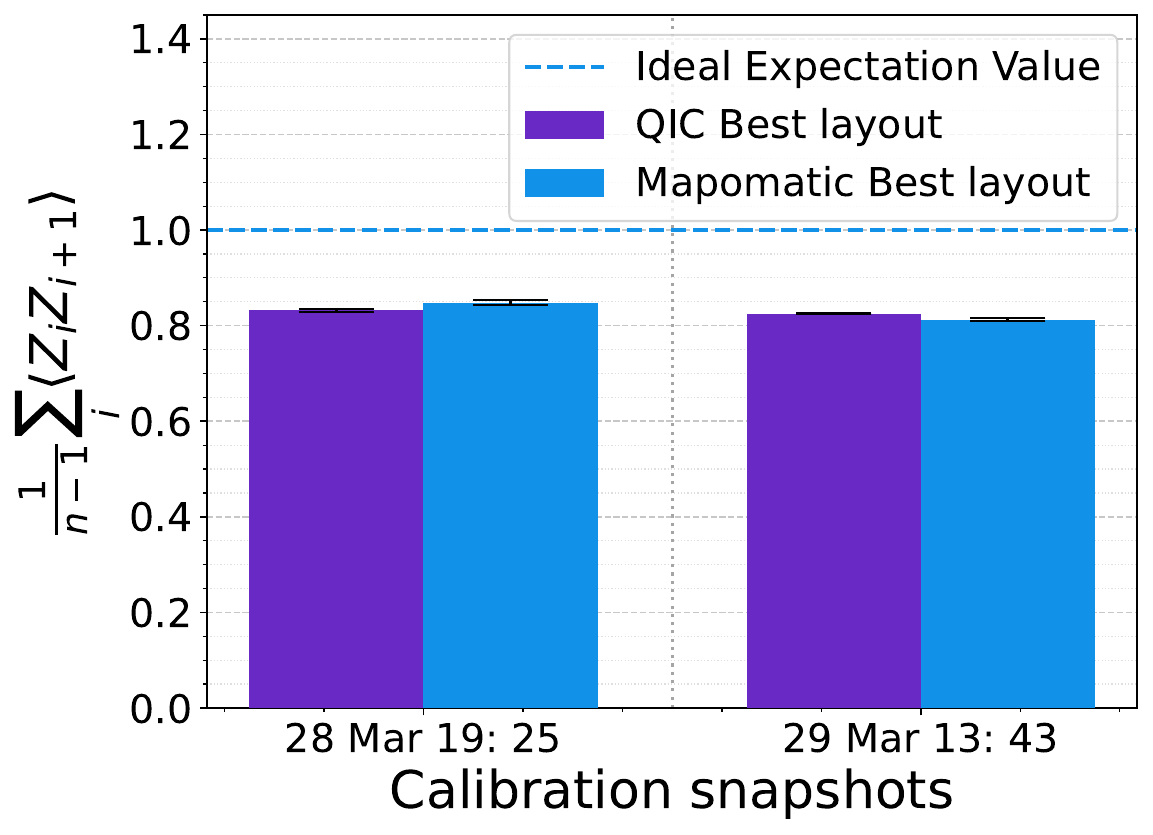}  
    \caption{
    Comparison between the best layouts selected by QIC and Mapomatic for a 12-qubit circuit at different time instances using an ibm\_fez-derived noise model, with Union QIC under a distortion threshold of $0$.
    }
    \label{fig: QIC_union_12q_ibm_fez}
\end{figure}

\begin{figure}[t]
    \centering        \includegraphics[width=0.5\columnwidth]{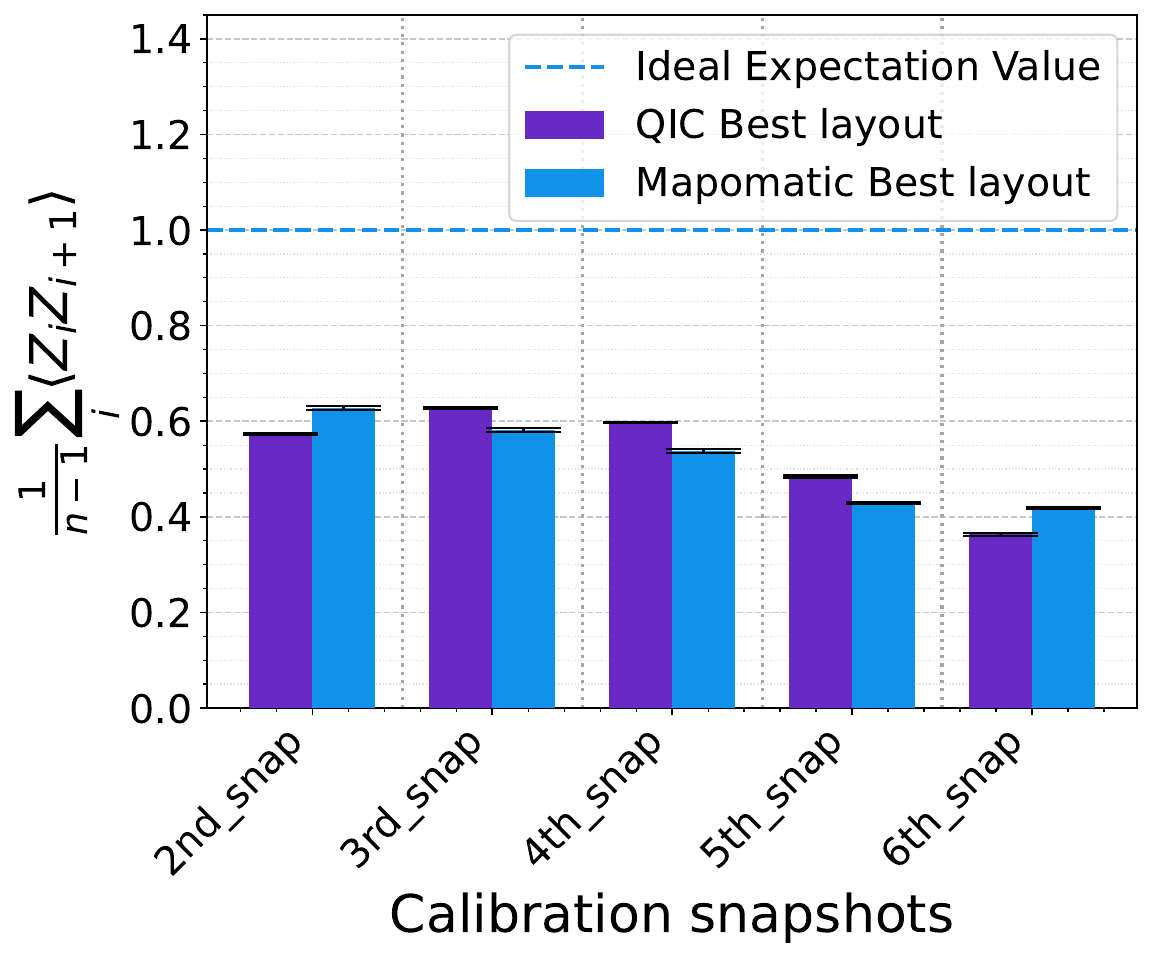}  
    \caption{
    Comparison between the best layouts selected by QIC and Mapomatic for a 12-qubit synthetic snapshot study derived from ibm\_fez, using Union QIC under a distortion threshold of $1$.
    }
    \label{fig: QIC_union_12q_snapshots_ibm_fez}
\end{figure}

In order to construct the minimal number of sets conforming to the distortion threshold, first it is necessary to construct the union QIC for a set of layouts. Algorithm~\ref{alg:construct_distortion}, first computes the required number of CNOT gates in the union QIC by taking a ceil of the average of number of CNOT gates across different relevant layouts corresponding to a physical qubit pair, where relevant layouts are those containing at least 1 CNOT gate corresponding to the physical qubit pair. This ceil average operation is done only for the relevant layouts since only the presence of CNOT gates on the corresponding physical qubit pair can introduce ambiguity in the CNOT count of the union QIC. Then, algorithm~\ref{alg:construct_distortion}
constructs the union QIC. After that, one can calculate the distortion faced by each layout in the union QIC, and hence determine whether all layouts within this union satisfies the provided distortion threshold or not. Algorithm~\ref{alg:distortion_compatibility_check} computes the total distortion faced by every layout  
and checks whether any of them exceed the given distortion threshold. If no layout exceeds the distortion threshold, then all the layouts are compatible with the union QIC.    

The core Algorithm~\ref{alg: greedy_algorithm_to_Reduce_qic_Execs_distortion} takes a list of isomorphic layouts and their corresponding QICs as input and computes the sets of layouts such that the layouts in each set may overlap, but maintaining the distortion threshold provided. For minimizing the number of sets, it follows the same greedy approach of Algorithm~\ref{alg:greedy} with one difference: instead of calling Algorithm~\ref{alg:disjointness} for checking compatibility, it instead calls Algorithm~\ref{alg:distortion_compatibility_check}.

\begin{lemma}
\label{alg6}
Algorithm~\ref{alg: greedy_algorithm_to_Reduce_qic_Execs_distortion} partitions the list of all $M$ isomorphic layouts into the minimal number of sets $U = \{S_1, S_2, \hdots, S_k\}$ with a time complexity of $\mathcal{O}(M^4n^2d^2)$ such that the layouts in each set may overlap but maintaining the distortion threshold provided, where $n$ is the number of qubits in a layout, and $d$ is the depth of a single QIC.
\end{lemma}

\begin{proof}
    See Appendix~\ref{proof6} .
\end{proof}

In the case of distortion, the effect of randomization is prominent. 
In Table~\ref{tab:permutation}, we show the number of union QIC sets, which is equal to the number of circuit executions, when the distortion threshold is set to 0, 1 and 2 for 20 random permutations. Note that 0 distortion threshold does not imply non-overlapping layouts; rather it means overlap such that the average number of 2-qubit gates associated with each qubit is maintained.
We note  
that the random permutation provides better outcomes as threshold increases. This further solidifies the requirement of random permutation and a quantum-centric supercomputing approach for finding the minimal set of layouts. In Table~\ref{tab:distortion_QIC}, we show the number of circuit executions required for Overlap QIC with a distortion threshold of 1 compared to the Union QIC with disjoint layouts, basic QIC and JIT. We note a significant reduction in the number of circuit executions for threshold 1, where on an average, we got relative reduction of $79.7\%$ in the number of circuit executions compared to JIT. This makes our proposed method extremely lightweight for near-term quantum computing.

\section{Comparison with Mapomatic}
\label{sec:result}
In this section, we present a comprehensive comparison between our proposed Quality Indicator Circuit (QIC) method and the established Mapomatic approach for layout selection in quantum devices. Our experiments span both real hardware and simulated environments, focusing on the practical impact of layout selection under varying noise profiles and circuit widths.

Our results demonstrate that the QIC method is capable of selecting layouts that yield improved expectation values compared to Mapomatic, particularly in dynamic noise settings. For instance, on a 127-qubit IBM quantum hardware, QIC with a distortion threshold of 1 was able to identify layouts for a 6-qubit $p=2$ QAOA circuit that produced higher-quality outcomes than those chosen by Mapomatic (see Fig.~\ref{fig: QIC_union_used_in_RH}). This improvement is especially notable given that Mapomatic relies on calibration data that may be outdated, whereas QIC provides real-time, circuit-aware scoring. Bootstrap sampling~\cite{10.1214/aos/1176344552} was used to ensure statistical robustness, and the observed gains are attributable solely to superior layout selection, as no error mitigation was applied.

\subsection{Scalability and Performance Across Circuit Widths}

As circuit width increases, the distinction between candidate layouts diminishes, and the performance gap between QIC and Mapomatic narrows. This trend is illustrated in Fig.~\ref{fig:qubit_scaling_qic}, where the best layouts selected by both methods converge for larger circuits. Such convergence is expected: with more qubits, the number of viable isomorphic layouts decreases, and the impact of layout selection alone becomes less pronounced without additional error mitigation. Nevertheless, for circuits of moderate width, QIC continues to offer tangible benefits in layout quality.

\subsection{Temporal Variation and Synthetic Snapshots}

To further probe the reliability of QIC under different noise regimes, we conducted snapshot experiments using both real and synthetic noise models derived from ibm\_fez. In slowly varying noise environments, the QIC-selected layout remains close to the Mapomatic-selected layout over time, with only modest variation in measured expectation values (see Fig.~\ref{fig: QIC_union_12q_ibm_fez}). This suggests that, when backend noise changes little, the relative ordering of candidate layouts is stable, and the advantage of real-time QIC-based selection may be limited.

In contrast, synthetic snapshot experiments with higher temporal variation reveal that QIC can often outperform mapomatic for some noise profiles (see Fig.~\ref{fig: QIC_union_12q_snapshots_ibm_fez}),
although there are some snapshots where Mapomatic prevails. The results underscore the generalizability of QIC to adapt to fast changing environments.

\subsection{Interpreting Limitations on Scalability and Reliability}

It is important to clarify that the occasional lack of decisive improvement over Mapomatic, especially in large circuits or slowly varying noise regimes, does not necessarily imply a fundamental scalability or performance issue with the QIC approach. Rather, these results reflect the inherent challenges of layout selection in quantum devices, where candidate layouts become less distinguishable as circuit size increases and noise profiles stabilise. This similar behaviour is observed for JIT as well. On the other hand, the substantial reduction in circuit execution overhead (up to 79.7\% compared to JIT, as shown in Table~\ref{tab:distortion_QIC}) and the ability to select better layouts in several cases demonstrate the viability of QIC. The observed limitations motivate further refinement, such as hybrid strategies that combine QIC with block-wise probing or machine learning-based corrections, but do not detract from its value. As observed in Fig.~\ref{fig: QIC_union_12q_snapshots_ibm_fez}, within its intended scope, real-time, circuit-aware layout selection for near-term devices, QIC remains a practical and scalable solution.


\section{Other Limitations and Future Extensibility}
\label{sec:limitations}

The QIC framework offers practical, real-time layout selection, but several limitations remain. First, QIC currently evaluates layouts within a single isomorphic family, not across multiple compilation choices. Extending the method to compare different layout families would enable broader optimization. Scalability is challenged as circuit width increases: fewer candidate layouts and less performance gap between them mean QIC and Mapomatic often converge for large circuits (see Fig.~\ref{fig:qubit_scaling_qic}). In slowly varying noise, QIC's advantage is reduced (see Fig.~\ref{fig: QIC_union_12q_ibm_fez}), but it remains effective for moderate circuits and dynamic environments.
While our experiments focus on ranking layouts within a single compiled family, the QIC scoring mechanism itself is not inherently limited to that setting. In principle, it can also be applied to candidate circuits generated by different compilation or routing choices, although such cross-family evaluation is left for future work.

Score fidelity is another area for improvement. The current QIC construction omits single-qubit gate errors, which can affect ranking reliability. Incorporating single-qubit effects or lightweight correction models would make QIC a more faithful proxy for execution quality. Synthetic snapshot experiments show some cases where QIC outperforms Mapomatic which motivate hybrid strategies and further refinement, rather than indicating fundamental limitations.

Future directions include cross-family optimization, block-wise probing for larger circuits, explicit modeling of single-qubit errors, and systematic analysis of overlap and distortion effects. Overall, QIC remains promising for near-term quantum devices, and ongoing research will further enhance its reliability and scalability.
\section{Conclusion}
\label{sec:conclusion}
In this article, we propose a lightweight technique for layout selection using QIC which, when executed on a layout, provides an indication of its noise profile. We show that the QIC method is targeted to each specific user circuit and layout, requires fewer circuit executions than JIT, and is comparable to Mapomatic in a static noise profile. Since the basic method requires computing a circuit for each isomorphic layout, which can be large, we propose the Union QIC method to unite QICs of multiple disjoint layouts and extend it to the Overlap QIC approach that allows some overlaps among the circuits while staying within a distortion threshold to maintain the reliability of layout scoring. Using this, we obtain a reduction in the number of circuit execution by $79.7\%$ when compared with JIT. 
QIC also demonstrates the capability to select better layouts than Mapomatic on a 127-qubit IBM quantum hardware in a dynamic noise setting.
At the same time, the larger-width and synthetic-snapshot studies indicate that this benefit is not uniform across all noise regimes, as discussed under limitations, with possible future directions also provided.
Overall, the substantial reduction in execution overhead, together with the ability to improve layout selection quality in dynamic settings, makes QIC a promising approach for near-term quantum computing.

\section*{Acknowledgments}
Authors S.S. and Y.S. were supported by grants from the IBM-IISc Hybrid Cloud lab.
The authors used AI-assisted tools only for language refinement and clarity. All ideas, methods, experiments, results and conclusions are the authors' own, and the final manuscript was reviewed and approved by the authors.


\bibliographystyle{IEEEtran}
\bibliography{arxiv}

@article{Peruzzo_2014,
   title={A variational eigenvalue solver on a photonic quantum processor},
   volume={5},
   ISSN={2041-1723},
   url={http://dx.doi.org/10.1038/ncomms5213},
   DOI={10.1038/ncomms5213},
   number={1},
   journal={Nature Communications},
   publisher={Springer Science and Business Media LLC},
   author={Peruzzo, Alberto and McClean, Jarrod and Shadbolt, Peter and Yung, Man-Hong and Zhou, Xiao-Qi and Love, Peter J. and Aspuru-Guzik, Alán and O’Brien, Jeremy L.},
   year={2014},
   month=jul }

@inproceedings{10.1145/3297858.3304007,
author = {Tannu, Swamit S. and Qureshi, Moinuddin K.},
title = {Not All Qubits Are Created Equal: A Case for Variability-Aware Policies for NISQ-Era Quantum Computers},
year = {2019},
isbn = {9781450362405},
publisher = {Association for Computing Machinery},
address = {New York, NY, USA},
url = {https://doi.org/10.1145/3297858.3304007},
doi = {10.1145/3297858.3304007},
booktitle = {Proceedings of the Twenty-Fourth International Conference on Architectural Support for Programming Languages and Operating Systems},
pages = {987–999},
numpages = {13},
location = {Providence, RI, USA},
series = {ASPLOS '19}
}

@article{lin2013optimized,
  title={Optimized quantum gate library for various physical machine descriptions},
  author={Lin, Chia-Chun and Chakrabarti, Amlan and Jha, Niraj K},
  journal={IEEE transactions on very large scale integration (VLSI) systems},
  volume={21},
  number={11},
  pages={2055--2068},
  year={2013},
  publisher={IEEE}
}

@article{lin2014paqcs,
  title={PAQCS: Physical design-aware fault-tolerant quantum circuit synthesis},
  author={Lin, Chia-Chun and Sur-Kolay, Susmita and Jha, Niraj K},
  journal={IEEE Transactions on Very Large Scale Integration (VLSI) Systems},
  volume={23},
  number={7},
  pages={1221--1234},
  year={2014},
  publisher={IEEE}
}

@article{silva2025hands,
  title={Hands-on introduction to randomized benchmarking},
  author={Silva, Ana and Greplova, Eliska},
  journal={SciPost Physics Lecture Notes},
  pages={097},
  year={2025}
}

@article{magesan2012characterizing,
  title={Characterizing quantum gates via randomized benchmarking},
  author={Magesan, Easwar and Gambetta, Jay M and Emerson, Joseph},
  journal={Physical Review A—Atomic, Molecular, and Optical Physics},
  volume={85},
  number={4},
  pages={042311},
  year={2012},
  publisher={APS}
}

@article{krantz2019quantum,
  title={A quantum engineer's guide to superconducting qubits},
  author={Krantz, Philip and Kjaergaard, Morten and Yan, Fei and Orlando, Terry P and Gustavsson, Simon and Oliver, William D},
  journal={Applied physics reviews},
  volume={6},
  number={2},
  year={2019},
  publisher={AIP Publishing}
}

@article{nation2021scalable,
  title={Scalable mitigation of measurement errors on quantum computers},
  author={Nation, Paul D and Kang, Hwajung and Sundaresan, Neereja and Gambetta, Jay M},
  journal={PRX Quantum},
  volume={2},
  number={4},
  pages={040326},
  year={2021},
  publisher={APS}
}

@article{zhu2022complexity,
  title={The complexity of quantum circuit mapping with fixed parameters},
  author={Zhu, Pengcheng and Zheng, Shenggen and Wei, Lihua and Cheng, Xueyun and Guan, Zhijin and Feng, Shiguang},
  journal={Quantum Information Processing},
  volume={21},
  number={10},
  pages={361},
  year={2022},
  publisher={Springer}
}

@article{zulehner2018efficient,
  title={An efficient methodology for mapping quantum circuits to the IBM QX architectures},
  author={Zulehner, Alwin and Paler, Alexandru and Wille, Robert},
  journal={IEEE Transactions on Computer-Aided Design of Integrated Circuits and Systems},
  volume={38},
  number={7},
  pages={1226--1236},
  year={2018},
  publisher={IEEE}
}

@inproceedings{murali2019noise,
  title={Noise-adaptive compiler mappings for noisy intermediate-scale quantum computers},
  author={Murali, Prakash and Baker, Jonathan M and Javadi-Abhari, Ali and Chong, Frederic T and Martonosi, Margaret},
  booktitle={Proceedings of the twenty-fourth international conference on architectural support for programming languages and operating systems},
  pages={1015--1029},
  year={2019}
}

@inproceedings{murali2020software,
  title={Software mitigation of crosstalk on noisy intermediate-scale quantum computers},
  author={Murali, Prakash and McKay, David C and Martonosi, Margaret and Javadi-Abhari, Ali},
  booktitle={Proceedings of the Twenty-Fifth International Conference on Architectural Support for Programming Languages and Operating Systems},
  pages={1001--1016},
  year={2020}
}

@inproceedings{niemann2016logic,
  title={Logic synthesis for quantum state generation},
  author={Niemann, Philipp and Datta, Rhitam and Wille, Robert},
  booktitle={2016 IEEE 46th International Symposium on Multiple-Valued Logic (ISMVL)},
  pages={247--252},
  year={2016},
  organization={IEEE}
}

@misc{murali2019noiseadaptivecompilermappingsnoisy,
      title={Noise-Adaptive Compiler Mappings for Noisy Intermediate-Scale Quantum Computers}, 
      author={Prakash Murali and Jonathan M. Baker and Ali Javadi Abhari and Frederic T. Chong and Margaret Martonosi},
      year={2019},
      eprint={1901.11054},
      archivePrefix={arXiv},
      primaryClass={quant-ph},
      url={https://arxiv.org/abs/1901.11054}, 
}

@misc{bhoumik2025distributedschedulingquantumcircuits,
      title={Distributed Scheduling of Quantum Circuits with Noise and Time Optimization}, 
      author={Debasmita Bhoumik and Ritajit Majumdar and Amit Saha and Susmita Sur-Kolay},
      year={2025},
      eprint={2309.06005},
      archivePrefix={arXiv},
      primaryClass={quant-ph},
      url={https://arxiv.org/abs/2309.06005}, 
}

@misc{shaik2023optimallayoutsynthesisquantum,
      title={Optimal Layout Synthesis for Quantum Circuits as Classical Planning (full version)}, 
      author={Irfansha Shaik and Jaco van de Pol},
      year={2023},
      eprint={2304.12014},
      archivePrefix={arXiv},
      primaryClass={quant-ph},
      url={https://arxiv.org/abs/2304.12014}, 
}

@misc{shaik2024optimallayoutsynthesisdeep,
      title={Optimal Layout Synthesis for Deep Quantum Circuits on NISQ Processors with 100+ Qubits}, 
      author={Irfansha Shaik and Jaco van de Pol},
      year={2024},
      eprint={2403.11598},
      archivePrefix={arXiv},
      primaryClass={quant-ph},
      url={https://arxiv.org/abs/2403.11598}, 
}

@INPROCEEDINGS{10247760,
  author={Lin, Wan-Hsuan and Kimko, Jason and Tan, Bochen and Bjørner, Nikolaj and Cong, Jason},
  booktitle={2023 60th ACM/IEEE Design Automation Conference (DAC)}, 
  title={Scalable Optimal Layout Synthesis for NISQ Quantum Processors}, 
  year={2023},
  volume={},
  number={},
  pages={1-6},
  doi={10.1109/DAC56929.2023.10247760}}

@book{nielsen2001quantum,
  title={Quantum computation and quantum information},
  author={Nielsen, Michael A and Chuang, Isaac L},
  volume={2},
  year={2001},
  publisher={Cambridge university press Cambridge}
}

@article{10.1214/aos/1176344552,
author = {B. Efron},
title = {{Bootstrap Methods: Another Look at the Jackknife}},
volume = {7},
journal = {The Annals of Statistics},
number = {1},
publisher = {Institute of Mathematical Statistics},
pages = {1 -- 26},
year = {1979},
doi = {10.1214/aos/1176344552},
URL = {https://doi.org/10.1214/aos/1176344552}
}

@online{fakebackend,
    title = {Fake backend simulation},    
    url = {https://quantum.cloud.ibm.com/docs/en/api/qiskit-ibm-runtime/fake-provider-fake-kolkata-v2},    
}

@book{Garey1979,
  author    = {Michael R. Garey and David S. Johnson},
  title     = {Computers and Intractability: A Guide to the Theory of NP-Completeness},
  year      = {1979},
  publisher = {W. H. Freeman},
  address   = {San Francisco}
}

@misc{qiskit-pattern,
    title = {Introduction to Qiskit patterns},
    url = {https://docs.quantum.ibm.com/guides/intro-to-patterns}
}

@article{nation2023suppressing,
  title={Suppressing quantum circuit errors due to system variability},
  author={Nation, Paul D and Treinish, Matthew},
  journal={PRX Quantum},
  volume={4},
  number={1},
  pages={010327},
  year={2023},
  publisher={APS}
}

@inproceedings{gambetta2024quantum,
  title={Quantum Software for the Utility-Scale and Beyond},
  author={Gambetta, Jay},
  booktitle={IEEE International Conference on Quantum Computing and Engineering},
  year={2024}
}

@article{farhi2014quantum,
  title={A quantum approximate optimization algorithm},
  author={Farhi, Edward and Goldstone, Jeffrey and Gutmann, Sam},
  journal={arXiv preprint arXiv:1411.4028},
  year={2014}
}

@article{bhoumik2024resource,
  title={Resource-aware scheduling of multiple quantum circuits on a hardware device},
  author={Bhoumik, Debasmita and Majumdar, Ritajit and Sur-Kolay, Susmita},
  journal={arXiv preprint arXiv:2407.08930},
  year={2024}
}

@inproceedings{li2019tackling,
  title={Tackling the qubit mapping problem for NISQ-era quantum devices},
  author={Li, Gushu and Ding, Yufei and Xie, Yuan},
  booktitle={Proceedings of the twenty-fourth international conference on architectural support for programming languages and operating systems},
  pages={1001--1014},
  year={2019}
}

@article{sivarajah2020t,
  title={t| ket>: a retargetable compiler for NISQ devices},
  author={Sivarajah, Seyon and Dilkes, Silas and Cowtan, Alexander and Simmons, Will and Edgington, Alec and Duncan, Ross},
  journal={Quantum Science and Technology},
  volume={6},
  number={1},
  pages={014003},
  year={2020},
  publisher={IOP Publishing}
}

@article{camps2020approximate,
  title={Approximate quantum circuit synthesis using block encodings},
  author={Camps, Daan and Van Beeumen, Roel},
  journal={Physical Review A},
  volume={102},
  number={5},
  pages={052411},
  year={2020},
  publisher={APS}
}

@article{amy2013meet,
  title={A meet-in-the-middle algorithm for fast synthesis of depth-optimal quantum circuits},
  author={Amy, Matthew and Maslov, Dmitri and Mosca, Michele and Roetteler, Martin},
  journal={IEEE Transactions on Computer-Aided Design of Integrated Circuits and Systems},
  volume={32},
  number={6},
  pages={818--830},
  year={2013},
  publisher={IEEE}
}

@online{ibm_qiskit_qubit_selection,
  title        = {Real-time benchmarking for qubit selection},
  howpublished = {\url{https://learning.quantum.ibm.com/tutorial/real-time-benchmarking-for-qubit-selection}},
}

@article{javadi2024quantum,
  title={Quantum computing with Qiskit},
  author={Javadi-Abhari, Ali and Treinish, Matthew and Krsulich, Kevin and Wood, Christopher J and Lishman, Jake and Gacon, Julien and Martiel, Simon and Nation, Paul D and Bishop, Lev S and Cross, Andrew W and others},
  journal={arXiv preprint arXiv:2405.08810},
  year={2024}
}

@misc{defensive_publication,
  author       = "{R. Majumdar, and P. V. Seshadri, and A. Ray}",
  title        = "{System and Method for Job Queue Management Using Joint Workload-Layout Quality Indicators}",
  year         = {2024},
  note         = {Defensive Publication},
  url          = {https://priorart.ip.com/IPCOM/00275214D},
  urldate      = {2024-10-18}
}

@misc{2405.13196,
Author = {David Kremer and Victor Villar and Hanhee Paik and Ivan Duran and Ismael Faro and Juan Cruz-Benito},
Title = {Practical and efficient quantum circuit synthesis and transpiling with Reinforcement Learning},
Year = {2024},
Eprint = {arXiv:2405.13196},
}

@article{zou2024lightsabre,
  title={LightSABRE: A Lightweight and Enhanced SABRE Algorithm},
  author={Zou, Henry and Treinish, Matthew and Hartman, Kevin and Ivrii, Alexander and Lishman, Jake},
  journal={arXiv preprint arXiv:2409.08368},
  year={2024}
}

@inproceedings{wilson2020just,
  title={Just-in-time quantum circuit transpilation reduces noise},
  author={Wilson, Ellis and Singh, Sudhakar and Mueller, Frank},
  booktitle={2020 IEEE international conference on quantum computing and engineering (QCE)},
  pages={345--355},
  year={2020},
  organization={IEEE}
}

@INPROCEEDINGS{9302814,
  author={Weder, Benjamin and Breitenbücher, Uwe and Leymann, Frank and Wild, Karoline},
  booktitle={2020 IEEE/ACM 13th International Conference on Utility and Cloud Computing (UCC)}, 
  title={Integrating Quantum Computing into Workflow Modeling and Execution}, 
  year={2020},
  volume={},
  number={},
  pages={279-291},
  doi={10.1109/UCC48980.2020.00046}}

@inproceedings{beisel2022metamodel,
  title={Metamodel and Formalization to Model, Transform, Deploy, and Execute Quantum Workflows},
  author={Beisel, Martin and Barzen, Johanna and Bechtold, Marvin and Leymann, Frank and Truger, Felix and Weder, Benjamin},
  booktitle={International Conference on Cloud Computing and Services Science},
  pages={113--136},
  year={2022},
  organization={Springer}
}

@article{10.1145/2431211.2431220,
author = {Saeedi, Mehdi and Markov, Igor L.},
title = {Synthesis and optimization of reversible circuits—a survey},
year = {2013},
issue_date = {February 2013},
publisher = {Association for Computing Machinery},
address = {New York, NY, USA},
volume = {45},
number = {2},
issn = {0360-0300},
url = {https://doi.org/10.1145/2431211.2431220},
doi = {10.1145/2431211.2431220},
journal = {ACM Comput. Surv.},
month = mar,
articleno = {21},
numpages = {34}
}

@inproceedings{10.1145/3316781.3317859,
author = {Wille, Robert and Burgholzer, Lukas and Zulehner, Alwin},
title = {Mapping Quantum Circuits to IBM QX Architectures Using the Minimal Number of SWAP and H Operations},
year = {2019},
isbn = {9781450367257},
publisher = {Association for Computing Machinery},
address = {New York, NY, USA},
url = {https://doi.org/10.1145/3316781.3317859},
doi = {10.1145/3316781.3317859},
booktitle = {Proceedings of the 56th Annual Design Automation Conference 2019},
articleno = {142},
numpages = {6},
location = {Las Vegas, NV, USA},
series = {DAC '19}
}


\clearpage
\begin{appendices}

\section{Proof of Lemma~\ref{lemma:qic_construction_time}}
\label{proof1}

\textbf{Lemma}~
\textit{Given a circuit of depth $d$ with $n$ qubits, Algorithm~\ref{alg:qic} constructs the corresponding QIC in $\mathcal{O}(n.d)$ time.}

\begin{proof}
For a given quantum circuit with $n$ qubits and depth $d$, Algorithm~\ref{alg:qic} scans through the circuit to find the indices of the 2-qubit gates, and the number of occurring for each index pair. This requires $\mathcal{O}(n.d)$ time since the circuit of depth $d$ with $n$ qubits can have at most $\mathcal{O}(n.d)$ number of quantum gates. After obtaining the indices of the 2-qubit operations, Algorithm~\ref{alg:qic} scans through this set once more to create a minimum ratio for the occurrence of gates for each index pair. Since the algorithm requires two scans of the entire circuit, therefore, the construction of the QIC can also be done with $\mathcal{O}(n.d)$ time.
\end{proof}

\section{Proof of Lemma~\ref{lemma:qic}}
\label{proof2}

\textbf{Lemma}~
    \textit{The ideal noise-free outcome of an $n$-qubit quantum circuit with a network of CNOT gates between arbitrary pairs of qubits sandwiched between two layers of Hadamard gates acting on all qubits is $|0\rangle ^{\otimes n}$.}

\begin{proof}
The circuit starts in the state $|0\rangle ^{\otimes n}$. The initial layer of Hadamard creates the uniform superposition state $\frac{1}{\sqrt{2^n}} \sum_{x=0}^{2^n-1} \ket{x}$. Now, consider a CNOT gate between any two pairs of qubits $j$ and $k$. This CNOT changes the state $\ket{x_0 x_1 ... x_j ... x_k ... x_{n-1}}$ to $\ket{x_0 x_1 ... x_j ... x_j \oplus x_k ... x_{n-1}}$. Similarly, since the state is a uniform superposition of all basis states, it also contains the state $\ket{x_0 x_1 ... x_j ... x_j \oplus x_k ... x_{n-1}}$ which gets transformed to $\ket{x_0 x_1 ... x_j ... x_k ... x_{n-1}}$. Thus effectively, there is no change in the uniform superposition. Hence, any network of CNOT gates retains the uniform superposition. Finally, applying a final layer of Hadamard gates at the end changes the equal superposition state to $|0\rangle ^{\otimes n}$.
\end{proof}

\section{Proof of Theorem~\ref{thm:chromatic}}
\label{proof3}
\textbf{Theorem}~
    \textit{Finding the minimum number of sets $S = \{S_1, S_2, ..., S_k\}$ such that $S_i \cap S_j = \phi$, $\forall$ sets $S_i \neq S_j$, and $\forall$ layouts $l^p, l^r \in S_i$, $\forall$ $i$, $l^p \cap l^r = \phi$ is NP-Hard.}

\begin{proof}
Assume that finding the minimum number of sets of isomorphic layouts, such that each set contains non-overlapping layouts, can be solved efficiently—that is, there exists a polynomial-time algorithm for this problem.  The
finding of the chromatic number of a general graph is known
to be NP-Hard \cite{Garey1979} 
which, as we show in the reduction below, can be transformed into finding the minimum number of sets of layouts where each set contains disjoint layouts in polynomial time.

\textbf{Reduction:}
Suppose we have a general graph with $N$ number of nodes and we want to find the chromatic number of the same.
We construct $N$ empty layouts, each assigned to a unique node in the given graph. If two nodes are connected by an edge, a common physical qubit will be assigned to their corresponding layouts as we describe below-\\
We start with the first node and identify all nodes it shares an edge with. Then, we assign physical qubit number 1 to each of these nodes, including the first node itself. This takes $\mathcal{O}(N-1)$ time. Remove the first node and its associated edges from the graph. Then, identify all nodes connected to the second node. Assign physical qubit number 2 to each of these nodes, including the second node itself. This takes $\mathcal{O}(N-2)$ time. Repeat this process until only one node remains in the graph. So, the total time taken will be $\mathcal{O}(N^2)$. Finally, make all layout sizes equal by assigning new qubit numbers — making sure that no new qubit number is repeated. To do this, we first check the lengths of all layouts in $\mathcal{O}(N^2)$ time, then identify the maximum length in $\mathcal{O}(N)$ time, and finally assign the remaining qubits in $\mathcal{O}(N^2)$ time. Overall, the process takes $\mathcal{O}(N^2)$ time. So, we’ve reduced the known NP-hard problem to our problem in polynomial time, assuming we have all-to-all qubit connectivity to make sure our layouts are isomorphic. In this reduction, any two isomorphic layouts share a common physical qubit if and only if their corresponding nodes in the graph are connected by an edge.

This reduction implies that if we had an efficient algorithm for our problem, we could use it as a subroutine to solve the chromatic number problem efficiently by assigning distinct colors to each set, which would mean that finding the chromatic number of a general graph is easy (not NP-hard). This leads to a contradiction, implying that such an efficient algorithm cannot exist.
Equivalently, reduction implies that since chromatic number problem is NP hard, our problem of interest must also be NP hard.
Hence, finding the minimum number of sets of isomorphic layouts, such that each set contains non-overlapping layouts, is NP-hard.
\end{proof}

\section{Proof of Lemma~\ref{alg2}}
\label{proof4}
\textbf{Lemma}~
    \textit{Algorithm~\ref{alg:disjointness} checks the compatibility of a layout $l$ with a set of mutually disjoint layouts $S = \{l^1, l^2, \hdots, l^m\}$ in $\mathcal{O}(n^2.m)$ time where $n$ is the number of qubits in each layout.}

\begin{proof}
For a given set of mutually disjoint layouts
$S = \{l^1, l^2, \hdots, l^m\}$ and a different layout $l$, where $n$ is the number of qubits in every layout, Algorithm~\ref{alg:disjointness} finds out whether the layout $l$ is disjoint with all layouts $\in S$ in $\mathcal{O}(n^2m)$ time since it checks whether any layout $\in S$ has overlap $> 0$ with the layout $l$ or not. Total number of overlap calculations required is $\mathcal{O}(m)$ since total number of layouts $\in S$ are $m$.
In one calculation of overlap between layouts $l$ and other layout $l^s \in S$ , $\mathcal{O}(n^2)$ time is required since every qubit $\in l$ needs to be compared with all qubits $\in l^s $.  
\end{proof}

\section{Proof of Lemma~\ref{alg3}}
\label{proof5}

\textbf{Lemma}~
    \textit{Algorithm~\ref{alg:greedy} partitions the list of all isomorphic layouts into a minimal number of sets $U = \{S_1, S_2, \hdots, S_k\}$ in $\mathcal{O}(M^3n^2)$ time where $M$ is the total number of isomorphic layouts and $n$ is the number of qubits in each layout.}

\begin{proof}
For $n$ qubits in each layout, algorithm~\ref{alg:greedy} scans the list of all $M$ isomorphic layouts and assigns each layout either to an existing compatible set or to a newly created set. This has time complexity $\mathcal{O}(M^3n^2)$ since it takes $\mathcal{O}(M)$ iterations for going through all layouts and assign them to one of the existing compatible set or to new set. Total number of sets will be $\mathcal{O}(M)$ and inside one set, number of layouts will also be $\mathcal{O}(M)$. From Lemma~\ref{alg2}, we know that each compatibility test requires $\mathcal{O}(Mn^2)$ time.
\end{proof}

\section{Proof of Lemma~\ref{alg6}}
\label{proof6}

\textbf{Lemma}~
\textit{Algorithm~\ref{alg: greedy_algorithm_to_Reduce_qic_Execs_distortion} partitions the list of all $M$ isomorphic layouts into a collection of mutually compatible sets $U = \{S_1, S_2, \hdots, S_k\}$ using a greedy heuristic, with a time complexity of $\mathcal{O}(M^4n^2d^2)$ such that the layouts in each set may overlap but maintaining the distortion threshold provided where $n$ is the number of qubits in a layout, and $d$ is the depth of a single QIC.}

\begin{proof}
Algorithm~\ref{alg: greedy_algorithm_to_Reduce_qic_Execs_distortion} uses algorithm~\ref{alg:distortion_compatibility_check} as a subroutine and algorithm~\ref{alg:distortion_compatibility_check} uses algorithm~\ref{alg:construct_distortion} as a subroutine.\par Algorithm~\ref{alg:construct_distortion} gets  $\mathcal{O}(M)$ layouts $\in S$ and $\mathcal{O}(M)$ corresponding $QICs$ each with depth $d$ as input and returns the corresponding Union QIC $U_Q$ in $\mathcal{O}(M^2n^2d^2)$ time since it takes $\mathcal{O}(M^2n^2)$ time to get the union of all qubits in the $QICs$ and $\mathcal{O}(M^2n^2d^2)$ time to get the union of all qubit pairs in the $QICs$ since total number of qubit pairs in a QIC can be at most $\mathcal{O}(nd)$. It takes $\mathcal{O}(nM)$ time to apply initial and last hadamard layer in $U_Q$. It iterates through $\mathcal{O}(ndM)$ qubit pairs for applying corresponding ceil average number of CNOT gates in $U_Q$ but ceil average number of CNOT gates for each qubit pair is calculated in $\mathcal{O}(M)$ time. Therefore, applying required number of CNOT gates in $U_Q$ can be done in $\mathcal{O}(ndM^2)$ time. Hence, overall, algorithm~\ref{alg:construct_distortion} takes $\mathcal{O}(M^2n^2d^2)$ time.\par 
Algorithm~\ref{alg:distortion_compatibility_check} finds out the distortion compatibility of a layout $l$ with a set of $\mathcal{O}(M)$ number of layouts $\in S$ in $\mathcal{O}(M^2n^2d^2)$ time. It first uses the algorithm~\ref{alg:construct_distortion} to construct the union QIC $U_Q$ in $\mathcal{O}(M^2n^2d^2)$ time. Then, it computes the distortion compatibility of the $ M+1$ layouts in $\mathcal{O}(Mn^2d)$ time where $\mathcal{O}(Mnd)$ time is required to get distortion corresponding to one qubit. Hence, overall algorithm~\ref{alg:distortion_compatibility_check} takes $\mathcal{O}(M^2n^2d^2)$ time.\par
For $n$ qubits in each layout, algorithm~\ref{alg: greedy_algorithm_to_Reduce_qic_Execs_distortion} partitions the list of all $M$ isomorphic layouts into a collection of distortion-compatible sets $U = \{S_1, S_2, \hdots, S_k\}$ in $\mathcal{O}(M^4n^2d^2)$ since it takes $\mathcal{O}(M)$ iterations for going through all layouts and assign them to one of the existing $\mathcal{O}(M)$ distortion compatible sets or to new set where distortion compatibility is check by algorithm~\ref{alg:distortion_compatibility_check} in $\mathcal{O}(M^2n^2d^2)$ time.
\end{proof}

\end{appendices}

\end{document}